\documentclass[10pt]{revtex4}

\usepackage{amsmath}
\usepackage{amssymb}
\usepackage{amsfonts}
\usepackage{amstext}
\usepackage{graphicx}
\usepackage{amscd}
\usepackage{mathrsfs}
\usepackage{stmaryrd}
\usepackage{slashed}
\usepackage{color}
\usepackage{upgreek}
\usepackage{tikz-cd}
\usepackage{hhline}
\usepackage{amsbsy}

\newcommand{\Sp}{\mathrm{Sp}}
\newcommand{\Ran}{\mathrm{Ran}\, }
\newcommand{\Orb}{\mathrm{Orb}}

\newcommand{\RE}{\Re\mathrm{e}}
\newcommand{\IM}{\Im\mathrm{m}}

\newcommand{\id}{\mathrm{id}}

\newcommand{\Aut}{\mathrm{Aut}}

\newcommand{\Lin}{\mathrm{Lin}}

\newcommand{\Ln}{\mathrm{Ln}\,}

\newtheorem{theo}{Theorem}

\newenvironment{proof}{\noindent \textit{Proof:}\small }{\normalsize \hfill $\Box$ \\}
\newenvironment{example}{\noindent \textbf{Example :} \hrulefill \\ \small}{\hrule \vspace{0.4cm}}

\begin{document}

\title{Ergodic and synthetic Koopman analyses of cat maps onto classical 2-tori}

\author{David Viennot}
\affiliation{Universit\'e Marie \& Louis Pasteur, Institut UTINAM (CNRS UMR 6213), Observatoire de Besan\c con, 41bis Avenue de l'Observatoire, BP1615, 25010 Besan\c con cedex, France.}

\begin{abstract}
  We study classical continuous automorphisms of the torus (cat maps) from the viewpoint of the Koopman theory. We find analytical formulae for Koopman modes defined coherently on the whole of the torus, and their decompositions associated  with the partition of the torus into ergodic components. The spectrum of the Koopman operator is studied in four cases of cat maps: cyclic, quasi-cyclic, critical (transition from quasi-cyclic to chaotic behaviour) and chaotic. The synthetic spectrum associated with the ergodic decomposition is also studied.
\end{abstract}

\maketitle


\section{Introduction}
Approaches based on eigenmodes are used in wave phenomena (classical electromagnetism, acoustics, fluid mechanics, etc). But this kind of approach seems inoperative for classical dynamical systems since classical flows are in general nonlinear. Nevertheless, Koopman and von Neumann have proposed another picture of the classical dynamics which is not based on phase space flows but on evolutions into the classical observable algebra \cite{Koopman1,Koopman2} (see also \cite{Lasota}). The ``evolution operator'' into the classical observable algebra, the so-called Koopman operator, is a linear operator permitting a spectral analysis of the dynamics. The properties of the Koopman operator depend on the chosen Banach space of phase space functions on which it is defined. In this paper we are interested by the case of the square integrable functions, because the spectral analysis is more consistent in the context of Hilbert spaces. By definition the spaces of observables of each ergodic components are stable by the Koopman operator. This one appears then as the ``synthesis'' of the operators for each ergodic space. In other words, the Koopman operator is ``block diagonal'' with respect to the ergodic decomposition, and is then the direct integral of its ergodic projections. An important point is that the nature of a spectral value (pure point or continuous) is not stable at the synthesis. Moreover an eigenmode defined by the ergodic decomposition has by construction a support restricted to an ergodic component, whereas the synthesis could permit to built eigenmodes coherently defined (for a physical meaning) on the whole of the phase space. With the ergodic decomposition, the eigenspaces are known from ``uncountable basis''. But these spaces are separable, we would like to be able to define countable basis having a physical meaning, such basis facilitating interpretations of the physical phenomena decomposed on it.\\

Interesting models of dynamical systems are the continuous automorphisms of the torus (cat maps). Following the Kolmogorov-Arnold-Moser (KAM) theorem \cite{Goldstein}, Hamiltonian systems present stable invariant tori in their phase spaces on which the dynamics is wrapped. A Poincar\'e section of that kind of dynamics becomes a discrete time dynamical system onto a torus of lower dimension. For example, a Hamiltonian system of two degrees of freedom obeying to the KAM theorem presents a stable 3-torus in the phase space $\mathbb R^4$ on which the dynamics is wrapped. The Poincar\'e section defined by $\theta^3=0$ (third angle of the 3-torus) is a 2-torus generated by all values of $(\theta^1,\theta^2,0)$. The dynamics is then represented by the Poincar\'e first recurrence map \cite{Lasota} onto this 2-torus. The model of a discrete time dynamics onto a torus is then relatively universal in classical dynamics. The cases of cat maps are the simplest non-trivial examples of this kind of systems presenting a large diversity of dynamical behaviours (cyclic, quasi-cyclic and chaotic).\\

Section II of this paper presents spectral analysis of the Koopman operators of cat maps. We present firstly an analysis based on the ergodic decomposition of the flow and of the Koopman operator. In a second step, we present an original construction of full Koopman modes of the ``synthetised operator'' which are coherently defined on the whole of the torus and which are directly related to the fixed and cyclic points of the system. This relation with special orbits provides an obvious interpretation of these eigenmodes. Section III studies in detail these full Koopman modes with respect to the dynamical behaviour of the cat map and presents examples for each one. For each cases, the synthetic and the ergodic decomposition spectra are studied.\\

{\it Some notations are used throughout this paper: $\Aut(M)$ denotes the set of automorphisms of a manifold $M$. $\mathcal L(\mathscr H)$ denotes the set of linear operators onto a Hilbert space $\mathscr H$ and $\mathcal U(\mathscr H)$ denotes the subset of unitary operators. $\Sp(A)$ denotes the spectrum of an operator $A$, $\Sp_{pp}(A)$ denotes its pure point spectrum and $\Sp_{cont}(A)$ its continuous spectrum. $A \simeq B$ (between two vector spaces) means that $A$ and $B$ are isomorphic. $U(1)$ denotes the unit circle in $\mathbb C$. $\Lin(S)$ denotes the linear envelope of the set $S$ (the set of all linear combinations of the elements of $S$)}

\section{Koopman analysis of classical cat maps}
\subsection{Cat maps and Koopman operators}
Let the phase space $\Gamma = \mathbb R^2/(2\pi\mathbb Z)^2 \simeq \mathbb S^1 \times \mathbb S^1$ be a 2-torus endowed with the uniform Borel measure $\mu$ such that $\mu(\Gamma)=1$ ($d\mu(\theta) = \frac{d\theta^1d\theta^2}{4\pi^2} = \frac{d^2\theta}{4\pi^2}$ with $\theta = (\theta^1,\theta^2)$ local coordinates onto $\Gamma$). A cat is a map $\varphi \in \Aut(\Gamma)$ defined by
\begin{equation}
  \varphi(\theta) = \Phi \theta \mod (2\pi,2\pi)
\end{equation}
where $\Phi \in SL(2,\mathbb R)$ is a $2\times 2$ real matrix with $\det \Phi=1$ and which is such that $\mu(\varphi(\Gamma))=1$. Strictly speaking, such a map is a cat if $\Phi \in SL(2,\mathbb Z)$ (in that case the condition $\mu(\varphi(\Gamma))=1$ is automatically ensured); but we enlarge the definition to obtain more dynamical behaviours. A cat map preserves the measure:
\begin{equation}
  \forall A \in \mathscr B(\Gamma), \quad \mu(\varphi(A)) = \mu(A)
\end{equation}
where $\mathscr B(\Gamma)$ is the Borel $\sigma$-algebra.\\
The nonlinearity of cat maps results from the modulo $(2\pi,2\pi)$. In practice, we can define a cat map more explicitly. We divide $\Gamma$ into several patches $\{D_i\}_i$ ($\bigcup_i D_i=\Gamma$, $D_i\cap D_j = \varnothing$ if $i\not=j$) such that $\exists m_i \in \mathbb Z^2$
\begin{equation}
  \varphi(\theta) = \Phi \theta + 2\pi m_i
\end{equation}
Let $i_\bullet : \Gamma \to \mathbb N$ be such that $\theta \in D_{i_\theta}$ ($\forall \theta \in \Gamma$).\\
Note that the condition $\mu(\varphi(\Gamma))=1$ induces that $\overline{\bigcup_i \varphi(D_i)} = \Gamma$ and $\mu(\varphi(D_i)\cap \varphi(D_j))=0$ if $i\not= j$.\\

Let $L^2(\Gamma,d\mu)$ be the algebra of classical square integrable observables of $\Gamma$. The Koopman operator $\mathcal K \in \mathcal L(L^2(\Gamma,d\mu))$ \cite{Koopman1,Koopman2,Lasota} is the evolution linear operator defined by:
\begin{equation}
  \forall f \in L^2(\Gamma,d\mu), \quad (\mathcal Kf)(\theta) = f(\varphi(\theta))
\end{equation}
The Koopman picture consists then to consider the evolution of the observables in place of the evolution of points of $\Gamma$ by $\varphi$. Since $\varphi$ is measure preserving, $\mathcal K \in \mathcal U(L^2(\Gamma,d\mu))$ is unitary.\\

We recall some known properties associated with $\Sp(\mathcal K)$:
\begin{itemize}
\item $\forall \varphi$, $1 \in \Sp_{pp}(\mathcal K)$ with at least the associated eigenvector $f_0(\theta)=1$ ($\forall \theta \in \Gamma)$.
\item $e^{\imath \lambda_1},e^{\imath \lambda_2} \in \Sp_{pp}(\mathcal K) \Rightarrow e^{\imath (\lambda_1+\lambda_2)} \in \Sp_{pp}(\mathcal K)$, an associated eigenvector being $f_{\lambda_1+\lambda_2} = f_{\lambda_1}f_{\lambda_2}$.
\item The flow $\varphi$ is ergodic if and only if the eigenvalues of $\mathcal K$ are simple and associated with unimodular eigenfunctions \cite{Eisner}. It follows that $\dim(1-\mathcal K)$ is the number of ergodic components of $\varphi$.
\item If there is no invariant subsets of $\Gamma$ of non-zero measure and if $\Sp(\mathcal K) \setminus \{1\}$ is absolutely continuous then the flow $\varphi$ is mixing \cite{Koopman2}.
\end{itemize}

Let $e^{\imath \lambda} \in \Sp(\mathcal K)$ be a spectral value of $\mathcal K$, we call Koopman eigenmode a function or a singular distribution $f_\lambda$ on $\Gamma$ such that
\begin{equation}\label{KoopEq}
  \mathcal K f_\lambda = e^{\imath \lambda} f_\lambda
\end{equation}
If $e^{\imath \lambda} \in \Sp_{pp}(\mathcal K)$, $f_\lambda \in L^2(\Gamma,d\mu)$ is the eigenvector (eigenfunction) associated with $e^{\imath \lambda}$. If $e^{\imath \lambda} \in \Sp_{cont}(\mathcal K)$, $f_\lambda \not\in L^2(\Gamma,d\mu)$ ($f_\lambda$ is a distribution or a non-normalisable function), and is the limit of a Weyl sequence in the weak topology \cite{Hislop}:\\
$\exists (f_{\lambda,n})_n \in (L^2(\Gamma,d\mu))^{\mathbb N}$, such that $\forall g \in L^2(\Gamma,d\mu)$, $(\langle f_{\lambda,n}|g\rangle)_n$ is a Cauchy sequence and $\lim_{n \to +\infty} \|(\mathcal K-e^{\imath \lambda})f_{\lambda,n}\|=0$; $\lim_{n \to +\infty} |\langle f_{n,\lambda}|g\rangle - \langle f_\lambda|g \rangle| = 0$, the first $\langle \bullet|\bullet \rangle$ being the inner product of $L^2(\Gamma,d\mu)$, the second one being the distribution bracket.\\
Remark about the used terminology: in order to avoid any confusion, we use the term ``Koopman eigenfunction'' only for functions of $L^2(\Gamma,d\mu)$ associated with pure point spectrum. To denote any solution of eq.(\ref{KoopEq}), including functions associated with the continuous spectrum (which are not $L^2$), we prefer to use the term ``Koopman mode'' with reference to Fourier modes for example. Strictly speaking \cite{Mezic},  for any function $g$ in the continuous spectrum eigenspace, we can write $g(\theta) = \int_{\Sp_{cont}(\mathcal K)} \langle f_\lambda|g \rangle f_\lambda(\theta) d\nu(\lambda)$ (for a suitable spectral measure onto $\Sp_{cont}(\mathcal K)$); the Koopman mode of $g$ associated with $\lambda$ is the projection of $g$: $\langle f_\lambda|g\rangle f_\lambda(\theta)$. So $f_\lambda$ can be viewed as a fundamental Koopman mode in this meaning. We forget the adjective ``fundamental'' to lighten the writing.\\ 

The goal of this section is the construction of these Koopman modes. The meaning of the Koopman modes are obvious, they are the (generalised) observables stable with respect to the flow. By ``stable'' we mean that up to a phase change (constant with $\theta$), the Koopman modes are stationary observables. The ergodic components of the flow $\varphi$ are by definition invariant submanifolds of $\Gamma$. We can then think that the Koopman modes are related to these ones.

\subsection{Ergodic decomposition and synthetic spectrum}
\subsubsection{Ergodic decomposition:}
Let $\theta_a \in \Gamma$ and $\Orb(\theta_a) = \{\varphi^n(\theta_a)\}_{n \in \mathbb Z}$ be its orbit by $\varphi$. By construction $\Gamma_a = \overline{\Orb(\theta_a)}$ (where the upperline stands for the topological closure with respect to $\mu$) is an ergodic component of the flow: this one travels the whole of $\Gamma_a$ and the orbits starting from points not belonging to $\Gamma_a$ do not cross $\Gamma_a$. We can then consider the decomposition of $\Gamma$ into ergodic components: it exists $\Omega \in \mathscr B(\Gamma)$ endowed with local coordinates $a$, such that $\{\Gamma_a = \overline{\Orb(\theta_a)}\}_{a \in \Omega}$ is a set of independent ergodic components covering $\Gamma$:
\begin{equation}
  \bigcup_{a \in \Omega} \Gamma_a = \Gamma \qquad \mu(\Gamma_a \cap \Gamma_b) = 0 \quad (\forall a\not=b)
\end{equation}
The choice of $\Omega$ is not unique; $\Omega$ is constituted by the choice of one point in each ergodic component $\Gamma_a$. If the ergodic components are discrete, $\Omega$ is 2 dimensional; if it is not the case, $\Omega$ is in general 1 dimensional, except if $\varphi$ is ergodic on the whole of $\Gamma$ where $\Omega$ is reduced to a single point. As far as possible, we choose the points constituting $\Omega$ such that this one be connected.\\
From the viewpoint of the Koopman picture, it exists a Borel measure $\nu$ onto $\Omega$, and a set mutually singular measures $\{\mu_a\}_{a\in \Omega}$ such that
\begin{equation}
  \int_\Omega \mu_a d\nu(a) = \mu \qquad \mu_a(\Gamma_b)=0 \quad (\forall a\not=b)
\end{equation}
The Hilbert space of observables is then decomposed as a fibre direct integral \cite{Reed, Chow} (equivalent to a direct sum of complementary subspaces but where the set of subspaces is uncountable) :
\begin{equation}
  L^2(\Gamma,d\mu) = \int_\Omega^\oplus L^2(\Gamma_a,d\mu_a) d\nu(a)
\end{equation}
If $\Gamma_a$ is discrete, $\mu_a$ is a discrete measure. By definition of an ergodic component, $\mathcal K$ is direct integral ``block diagonal'' (it induces no coupling between two ergodic components):
\begin{equation}\label{fibreint}
  \mathcal K = \int_\Omega^\oplus \left. \mathcal K \right|_{\Gamma_a} d\nu(a)
\end{equation}
where $\left. \mathcal K \right|_{\Gamma_a}$ denotes the Koopman operator of the flow $\varphi$ restricted on $\Gamma_a$, which is also the projection of $\mathcal K$ onto the Hilbert subspace $L^2(\Gamma_a,d\mu_a)$:
\begin{equation}
  \left. \mathcal K \right|_{\Gamma_a} = \mathcal K \mathcal P_a = \mathcal P_a \mathcal K = \mathcal P_a \mathcal K \mathcal P_a
\end{equation}
where $\mathcal P_a$ is the orthogonal projection from $L^2(\Gamma,d\mu)$ to $L^2(\Gamma_a,d\mu_a)$. Eq.(\ref{fibreint}) can be also interpreted as follows: $\left. \mathcal K \right|_{\Gamma_a}$ are the independent Koopman operators of the independent flows of the ergodic components, whereas $\mathcal K$ is the ``synthetised'' operator obtained by sum of them.

\subsubsection{Koopman modes on a single ergodic component:}
We consider the spectral analysis on a single ergodic component $\Gamma_a$.
\begin{theo} \label{ThKoopModeErgo}
  If $\theta_a$ is $N$-cyclic ($N \in \mathbb N^*$), i.e. $\varphi^N(\theta_a)=\theta_a \iff \Gamma_a = \Orb(\theta_a) = \{\varphi^n(\theta_a)\}_{n\in\{0,...,N-1\}}$, then $\Sp(\left. \mathcal K \right|_{\Gamma_a}) = \Sp_{pp} (\left. \mathcal K \right|_{\Gamma_a}) = \{e^{\imath \frac{2\pi}{N} p}\}_{p\in\{0,...,N-1\}}$ and the associated eigenvectors are
  \begin{equation}
  f_p = \sum_{n=0}^{N-1} e^{\imath \frac{2\pi}{N} np} \delta_{\varphi^n(\theta_a)}
\end{equation}
  where $f_p \equiv f_{\frac{2\pi}{N} p}$ and $\delta_{\varphi^n(\theta_a)} (\theta) = \delta(\theta-\varphi^n(\theta_a))$. In that case $\mu_a$ is discrete and $L^2(\Gamma_a,d\mu_a) \simeq \mathbb C^N$ via the canonical basis $(\delta_{\varphi^n(\theta_a)})_{n\in\{0,...,N-1\}}$.\\

  If $\theta_a$ is not cyclic then $\Sp(\left. \mathcal K \right|_{\Gamma_a}) = U(1)$ and the Koopman modes are
  \begin{equation}
    f_\lambda = \sum_{n \in \mathbb Z} e^{\imath \lambda n} \delta_{\varphi^n(\theta_a)}
  \end{equation}
\end{theo}

\begin{proof}
  Let $\Gamma_a = \Orb(\theta_a) = \{\varphi^n(\theta_a)\}_{n\in\{0,...,N-1\}}$ be the orbit of a $N$-cyclic point.
\begin{eqnarray}
  \left. \mathcal K \right|_{\Gamma_a} f_p(\theta) & = & f_p(\varphi(\theta)) \\
  & = & \sum_{n=0}^{N-1} e^{\imath \frac{2\pi}{N} np} \delta_{\varphi^{n-1}(\theta_a)}(\theta) \\
  & = & \sum_{n=-1}^{N-2} e^{\imath \frac{2\pi}{N} (n+1)p} \delta_{\varphi^n(\theta_a)}(\theta) \\
  & = & \sum_{n=0}^{N-2} e^{\imath \frac{2\pi}{N} (n+1)p} \delta_{\varphi^n(\theta_a)}(\theta) + \delta_{\varphi^{N-1}(\theta_a)}(\theta) \\
  & = & \sum_{n=0}^{N-2} e^{\imath \frac{2\pi}{N} (n+1)p} \delta_{\varphi^n(\theta_a)}(\theta) + e^{\imath \frac{2\pi}{N} Np}\delta_{\varphi^{N-1}(\theta_a)}(\theta) \\
  & = & e^{\imath \frac{2\pi}{N} p} f_p(\theta)
\end{eqnarray}
We see that if $\lambda$ is not a $N$-th root of unity, it cannot be a Koopman value of $\left. \mathcal K \right|_{\Gamma_a}$.\\

Let $\Gamma_a = \overline{\Orb(\theta_a)}$ where $\theta_a$ is not cyclic.
\begin{eqnarray}
  \left. \mathcal K \right|_{\Gamma_a} f_\lambda(\theta) & = & f_\lambda(\varphi(\theta)) \\
  & = & \sum_{n=-\infty}^{+\infty} e^{\imath \lambda n} \delta_{\varphi^{n-1}(\theta_a)}(\theta) \\
  & = & \sum_{n=-\infty}^{+\infty} e^{\imath \lambda (n+1)} \delta_{\varphi^n(\theta_a)}(\theta) \\
  & = & e^{\imath \lambda} f_p(\theta)
\end{eqnarray}
No restriction occurs concerning the Koopman value $\lambda$.
\end{proof}

In the not cyclic case, the nature of $e^{\imath \lambda}$ (eigenvalue of the pure point spectrum or spectral value of the continuum) depends on the dynamical behaviour of $\varphi$. We will treat this question in section III.

\subsubsection{Synthetic spectrum:}
$\mathcal K$ being normal (since it is unitary) we have $\Sp(\mathcal K) = \overline{\bigcup_{a \in \Omega} \Sp(\left. \mathcal K \right|_{\Gamma_a})}$ (see \cite{Cho}). Let $e^{\imath \lambda} \in \bigcup_{a \in \Omega} \Sp(\left. \mathcal K \right|_{\Gamma_a})$ and $\Omega_\lambda \in \mathscr B(\Omega)$ be its support, i.e. $\Omega_\lambda = \{a \in \Omega, \text{ s.t. } e^{\imath \lambda} \in \Sp(\left. \mathcal K \right|_{\Gamma_a})\}$. $\Sp_{synth}(\mathcal K) = \{e^{\imath \lambda} \in \bigcup_{a \in \Omega} \Sp(\left. \mathcal K \right|_{\Gamma_a}), \text{ such that } \nu(\Omega_\lambda) \not= 0\}$ is called the synthetic spectrum of $\mathcal K$ \cite{Azoff}. We see easily that:
\begin{itemize}
\item If $e^{\imath \lambda}$ is a synthetic pure point eigenvalue ($\nu(\Omega_\lambda) \not= 0$ and $e^{\imath \lambda} \in \Sp_{pp}(\left. \mathcal K \right|_{\Gamma_a})$ $\forall a \in \Omega_\lambda$), then $e^{\imath \lambda} \in \Sp_{pp}(\mathcal K)$.
\item If $e^{\imath \lambda}$ is a non-synthetic pure point eigenvalue ($\nu(\Omega_\lambda) = 0$ and $e^{\imath \lambda} \in \Sp_{pp}(\left. \mathcal K \right|_{\Gamma_a})$ $\forall a \in \Omega_\lambda$), then $e^{\imath \lambda} \in \Sp_{cont}(\mathcal K)$.
\item If $e^{\imath \lambda}$ is a non-synthetic spectral value of the continuum ($\nu(\Omega_\lambda) = 0$ and $e^{\imath \lambda} \in \Sp_{cont}(\left. \mathcal K \right|_{\Gamma_a})$ $\forall a \in \Omega_\lambda$), then $e^{\imath \lambda} \in \Sp_{cont}(\mathcal K)$.
\end{itemize}

If $e^{\imath \lambda}$ is a synthetic spectral value of the continuum, then $\Ran(e^{\imath \lambda}-\left. \mathcal K \right|_{\Gamma_a}) \not= L^2(\Gamma_a,d\mu_a)$ and $\ker(e^{\imath \lambda}-\left. \mathcal K \right|_{\Gamma_a}) = \{0\}$. But nothing ensures that $\ker(e^{\imath \lambda} - \mathcal K)=\{0\}$. Indeed, we can have $(e^{\imath \lambda} - \left. \mathcal K \right|_{\Gamma_a}) f_{\lambda a} \not=0$ for $\nu$-almost all $a \in \Omega_\lambda$, but with $\int_{\Omega_\lambda} (e^{\imath \lambda} - \left. \mathcal K \right|_{\Gamma_a}) f_{\lambda a} d\nu(a) = 0$, for example if $(e^{\imath \lambda} - \left. \mathcal K \right|_{\Gamma_a}) f_{\lambda a}= \cos(a)$, $\Omega_\lambda = [0,2\pi]$ and $d\nu(a)=da$. In this case $f_{\lambda a} \not\in \ker(e^{\imath \lambda}-\left. \mathcal K \right|_{\Gamma_a})$ but $\int_{\Omega}^\oplus f_{\lambda a} d\nu(a) \in \ker(e^{\imath \lambda} - \mathcal K)$. We have then no conclusion, $\lambda$ can be a pure point eigenvalue of $\mathcal K$ or a spectral value of the continuum of $\mathcal K$.\\
To summarize, the synthesis of the ergodic projections $\int_\Omega^\oplus \left. \mathcal K \right|_{\Gamma_a} d\mu_a$ tends to weaken the nature of the spectral values (from pure point to continuous nature). But by definition, the synthetic values carry significant weight in the synthesis. It follows that the synthetic pure point eigenvalues remain pure point. Moreover, ``synthesis resonances'' can transform some synthetic spectral values of the continuum to pure point eigenvalues.\\

$\forall e^{\imath \lambda} \in \bigcup_{a \in \Omega} \Sp(\left. \mathcal K \right|_{\Gamma_a})$ we can define Koopman modes on the whole of $\Gamma$ by example as
\begin{equation}
  f_\lambda = \int_{\Omega_\lambda}^\oplus \sum_{i=1}^{d_{\lambda a}} c_{ia} f_{\lambda i a} d\nu(a)
\end{equation}
where $i$ is a possible degeneracy index of $e^{\imath \lambda}$ in $\Sp(\left. \mathcal K \right|_{\Gamma_a})$ and $c_{ia} \in \mathbb C$. By construction $f_\lambda(\theta)=0$, $\forall \theta \not\in \bigcup_{a \in \Omega_\lambda} \Gamma_a$. If $e^{\imath \lambda}$ is a synthetic pure point eigenvalue, we must choose $a \mapsto c_{ia}$ such that $\int_{\Omega_\lambda} \sum_i |c_{ia}|^2 d\nu(a) < +\infty$ in order to define an associated eigenvector. But in the case where $e^{\imath \lambda}$ is a synthetic spectral value of the continuum becoming a pure point eigenvalue of $\mathcal K$, it is not clear how choose $c_{ia}$ in order to $f_\lambda \in L^2(\Gamma,d\mu)$. Moreover, in all cases, such Koopman modes defined with arbitrary $c_{ia}$ are just superpositions of the Koopman modes associated with the ergodic components without other meaning. We might want to define Koopman on the whole of $\Gamma$ characterising global properties of the flow. Moreover the eigenspace $\Lin(f_{\lambda i a})_{i \in \{1,...,d_{\lambda a}\}, a \in \Omega_\lambda}$ is known via a non-$L^2$ ``uncountable basis'' $(f_{\lambda i a})_{i \in \{1,...,d_{\lambda a}\}, a \in \Omega_\lambda}$ of functions/distributions having a $\mu$-null support $\Gamma_a$ ($\mu(\Gamma_a)=0$ except if $\Omega$ is countable). But for a pure point eigenvalue the eigensubspace $\ker(e^{\imath \lambda}-\mathcal K) \subset L^2(\Gamma,d\mu)$ is separable, it has $L^2$ countable basis of functions of non-null supports. Note that if $\lambda \in \Sp_{cont}(\mathcal K)$, the generalised eigenspace $\Lin(f_{\lambda i a})_{i,a} \not\subset L^2(\Gamma,d\mu)$ is not $\ker(e^{\imath \lambda}-\mathcal K)$ (the kernel being defined as a subspace of the Hilbert space $L^2(\Gamma,d\mu)$). Finally we want to find Koopman modes defined coherently on the whole of $\Gamma$ characterizing global behaviour of the flow and forming countable basis of the eigenspaces at least for the pure point eigenvalues. This is the subject of the next subsection.

\subsection{Full Koopman modes}
For a continuous time periodic dynamical system (harmonic oscillator), the Koopman modes reduce to the Fourier modes. We can think the Koopman analysis as being a kind of generalisation of the Fourier analysis, and so the Koopman modes are related to the (quasi)cyclicality of the flow. As ``stable'' observables, it is natural that the Koopman modes reflect the (quasi)-cyclic properties of the system.\\

The Koopman modes are then strongly related to fixed or cyclic points of the cat maps, and then to the spectral analysis of $\Phi$. Let $\Sp(\Phi) = \{e^{\pm \imath \omega}\}$ be the spectrum of $\Phi$ (supposed diagonalisable, the non-diagonalisable case will be treated section III) and $e_\pm$ and $e^\pm$ be the associated right and left eigenvectors
\begin{equation}
  {\Phi^\mu}_\nu e^\nu_\pm = e^{\pm \imath \omega} e^\mu_\pm \qquad e^\pm_\mu {\Phi^\mu}_\nu = e^{\pm \imath \omega} e^\pm_\nu
\end{equation}
and let $\theta^\pm = e^\pm_\mu \theta^\mu$ (where $e^a_\mu e^\mu_b = \delta^a_b$) be the complex variables in the eigendirections (since $\Phi$ is not necessarily symmetric, the eigenvalues and eigenvectors can be complex). In general we must then consider the complexification of the torus $\Gamma_{\mathbb C}$ (the complex manifold generated by $\theta^\pm$).\\
In the patch $D_i$, the fixed points satisfies:
\begin{eqnarray}
  \phi_i^\pm & = & \varphi(\phi_i^\pm) \\
  & = & e_\mu^\pm ({\Phi^\mu}_\nu \phi_{i}^\nu + 2\pi m_{i}^\mu) \\
  & = & e^{\pm \imath \omega} \phi_{i}^\pm + 2\pi m_{i}^\pm
\end{eqnarray}
It follows that the complex fixed points of $\varphi$ are (in addition to the trivial fixed point $0$)
\begin{equation}
  \phi_{i}^\pm = \frac{2\pi e^\pm_\mu m^\mu_{i}}{1-e^{\pm \imath \omega}}
\end{equation}
For the $N$-cyclic points, the situation is more difficult since we must follow all possible changes of patches during the cycle. $\mu(\Gamma)=1$ and $\varphi$ being measure-preserving, the Poincar\'e recurrence theorem \cite{Lasota} ensures that the orbit of $\mu$-almost all $\theta$ returns in the neighbourhood of $\theta$ after a certain number of iterations. This induces that $\exists N \in \mathbb N^*$ such that $i_{\varphi^N(\theta)} = i_{\theta}$ (the orbit starting in $D_{i_{\theta}}$ returns necessarily into $D_{i_{\theta}}$ for almost all $\theta$). We can use points $\theta$ of $D_{i_\theta}$ to follow the changes of patches, by searching associated cyclic points $\phi_{i_\theta}$. We have then $\phi_{i_{\varphi(\theta)}} = \varphi(\phi_{i_\theta}) \Rightarrow \phi_{i_{\varphi^2(\theta)}} = \varphi(\phi_{i_{\varphi(\theta)}}) = \varphi^2(\phi_{i_\theta}) \Rightarrow ... \Rightarrow \phi_{i_{\varphi^N(\theta)}} = \varphi^N(\phi_{i_\theta})$. Since $i_{\varphi^N(\theta)} = i_\theta$ we must have $\varphi^N(\phi_{i_\theta}) = \phi_{i_\theta}$, $\phi_{i_\theta}$ is well a $N$-cyclic complex point of $\varphi$. This one is then defined by
\begin{eqnarray}
  \phi_{i_\theta}^\pm  & = & \varphi^N(\phi_{i_\theta})^\pm \\
  & = & e^\pm_\mu ({\Phi^\mu}_{\nu_1}({\Phi^{\nu_1}}_{\nu_2}(...({\Phi^{\nu_{N-1}}}_{\nu_N} \phi^{\nu_N}_{i_\theta} \nonumber \\
  & & \qquad +2\pi m^{\nu_{N-1}}_{i_\theta})+...)+2\pi m^{\nu_1}_{i_{\varphi^{N-2}(\theta)}})+2\pi m^\mu_{i_{\varphi^{N-1}(\theta)}}) \\
  & = & e^{\pm \imath N\omega} \phi^\pm_{i_\theta} + 2\pi e^{\pm \imath (N-1)\omega} m^\pm_{i_\theta}+... \nonumber \\
  & & \quad ... +2 \pi e^{\pm \imath \omega} m^\pm_{i_{\varphi^{N-2}(\theta)}} + 2\pi m^\pm_{i_{\varphi^{N-1}(\theta)}}
\end{eqnarray}
It follows that the complex $N$-cyclic points of $\varphi$ are (if $N\omega \not \in 2\pi \mathbb Z$)
\begin{equation}
  \phi_{i_\theta}^\pm = \frac{2\pi e^\pm_\mu \sum_{q=0}^{N-1} e^{\pm \imath (N-1-q)\omega} m^\mu_{i_{\varphi^q(\theta)}}}{1-e^{\pm \imath N \omega}}
\end{equation}
Note that we have by construction
\begin{equation}
  \varphi(\theta)^\pm - \phi^\pm_{i_{\varphi(\theta)}} = e^{\pm \imath \omega}(\theta^\pm - \phi^\pm_{i_{\theta}})
 \end{equation}

If $\omega \in \imath \mathbb R$ these fixed and cyclic points are obviously real and are then true fixed and cyclic points of the cat map onto $\Gamma$. If $\omega \in \mathbb R$, $\phi^+_{i_\theta} = \overline{\phi^-_{i_\theta}}$ and then $\phi_{i_\theta} \in \Gamma$ solution of $e^+_\mu \phi_{i_\theta}^\mu = \phi_{i_\theta}^+$ are also true fixed or cyclic points.\\

We can now built Koopman modes defined on the whole of $\Gamma$ and related to the (complex) fixed and cyclic points of $\varphi$:
\begin{theo}\label{ThFullKoopMode}
  If $\omega \in \mathbb R$, then $\Sp(\mathcal K) = \overline{\{e^{\imath n \omega}\}_{n \in \mathbb Z}}$ and the associated Koopman modes are 
  \begin{equation}\label{solcat}
    f_{nml}(\theta) = |\varphi^l(\theta)^+ - \phi_{i_{\varphi^l(\theta)}}^+|^m e^{\imath n \arg(\varphi^l(\theta)^+ - \phi_{i_{\varphi^l(\theta)}}^+)}
  \end{equation}
  with $e^{\imath \lambda_n}= e^{\imath n \omega}$ generally of infinite multiplicity with $l,m \in \mathbb N$ (indexes of degeneracy). $\phi_{i_\theta}^+$ is the complex representation of a $N_\theta$-cyclic point (fixed point if $N_\theta=1$) of $\varphi$ associated with the sequence $\{m_{i_{\varphi^q(\theta)}}\}_{q \in \{0,...,N_\theta\}}$ , $N_\theta \in \mathbb N^*$ being the time of the first return of the orbit of $\theta$ in $D_{i_\theta}$.\\

  If $\omega \in \imath \mathbb R$, then $\Sp(\mathcal K) = U(1)$ and the Koopman modes associated with $e^{\imath \lambda}$ are
   \begin{equation}
    f_{\lambda\alpha l}(\theta)  =  e^{-\imath \frac{\lambda}{\IM(\omega)} \left(\alpha \ln (\varphi^l(\theta)^+-\phi^+_{i_{\varphi^l(\theta)}}) - (1-\alpha) \ln (\varphi^l(\theta)^--\phi^-_{i_{\varphi^l(\theta)}})\right)}
   \end{equation}
   with $\alpha \in \mathbb R$ and $l \in \mathbb N$.
\end{theo}

\begin{proof}
To solve the eigenequation of $\mathcal K$, we search the eigenfunctions as $f_\lambda(\theta) = e^{\pm \frac{\lambda}{\omega}F(e^\pm_\mu \theta^\mu - \phi_{\psi(\theta)}^\pm)}$, i.e. as functions of the complexified coordinates $\theta^\pm$. This is just a general parametrisation, the function $F \in \mathbb C^{\mathbb C}$, the complex shift $\phi^\pm_\theta$ and the measure-preserving automorphism $\psi \in \Aut(\Gamma)$ being undetermined this choice does not induce any loss of generality. Firstly, we search the solutions such that $\psi = \id_\Gamma$.
\begin{eqnarray}
  & & \mathcal K f_\lambda = e^{\imath \lambda} f_\lambda \nonumber \\
  & \iff & e^{\pm \frac{\lambda}{\omega}F(e^\pm_\mu \varphi(\theta)^\mu - \phi_{\varphi(\theta)}^\pm)} = e^{\imath \lambda} e^{\pm \frac{\lambda}{\omega}F(e^\pm_\mu \theta^\mu - \phi_\theta^\pm)} \\
  & \iff & F(e^\pm_\mu \varphi(\theta)^\mu - \phi_{\varphi(\theta)}^\pm) = F(e^\pm_\mu \theta^\mu - \phi_\theta^\pm) \pm \imath \omega \\
  & \iff & F(e^\pm_\mu ({\Phi^\mu}_\nu \theta^\nu + 2\pi m^\nu_{i_\theta}) -\phi_{\varphi(\theta)}^\pm) = F(e^\pm_\mu \theta^\mu - \phi_\theta^\pm) \pm \imath \omega \\
  & \iff & F(e^{\pm \imath \omega} \theta^\pm + 2\pi m^\pm_{i_\theta} - \phi_{\varphi(\theta)}^\pm) = F(\theta^\pm - \phi_\theta^\pm) \pm \imath \omega \\
  & \iff & F(e^{\pm \imath \omega} (\theta^\pm + e^{\mp \imath \omega}(2\pi m^\pm_{i_\theta} - \phi_{\varphi(\theta)}^\pm)))  = F(\theta^\pm - \phi_\theta^\pm) \pm \imath \omega \\
  & \iff & \left\{\begin{array}{c} F(z) = \Ln(z)+C = \ln|z|+\imath \arg(z) + C \\ 2\pi m^\pm_{i_\theta} - \phi_{\varphi(\theta)}^\pm = - e^{\pm \imath \omega} \phi_\theta^\pm \end{array} \right.
\end{eqnarray}
with $C$ a constant. The equation governing the shift $\phi^\pm_\theta$ depending on $\theta$ only from the index $i_\theta$, from this point we write $\phi^\pm_{i_\theta}$. We can rewrite this equation as
\begin{eqnarray}
  \phi_{i_{\varphi(\theta)}}^\pm & = & e^{\pm \imath \omega} \phi_{i_\theta}^\pm + 2\pi m_{i_\theta}^\pm \\
  & = & e_\mu^\pm ({\Phi^\mu}_\nu \phi_{i_\theta}^\nu + 2\pi m_{i_\theta}^\mu) \\
  & = & \varphi(\phi_{i_{\theta}})^\pm
\end{eqnarray}

This is the equation defining a complex fixed or cyclic point of $\varphi$.

\begin{equation}
  \varphi(\theta)^\pm - \phi^\pm_{i_{\varphi(\theta)}} = e^{\pm \imath \omega}(\theta^\pm - \phi^\pm_{i_{\theta}})
 \end{equation}
The nature of $\phi_{i_\theta}$ depending on the time $N_\theta$ of the first return of $\theta$ in $D_{i_{\theta}}$.\\
Now we consider the case of a non-trivial $\psi \in \Aut(\Gamma)$ defined by $\psi(\theta) = \Psi \theta + 2\pi \ell_{i_\theta}$ with $\Psi \in SL(2,\mathbb R)$ and $\ell_{i_\theta} \in \mathbb Z^2$. We have then
\begin{equation}
  e^\pm_\mu \psi(\varphi(\theta))^\mu  = e^\pm_\mu({\Psi^\mu}_\nu({\Phi^\mu}_\rho \theta^\rho+2\pi m^\rho_{i_\theta})+2\pi \ell^\mu_{i_{\varphi(\theta)}})
\end{equation}
To find anew an eigenequation with $e^{\pm \imath \omega}$ it is necessary that $e^\pm_\mu$ be also the left eigenvectors of $\Psi$, implying the following condition on $\psi$: $[\Psi,\Phi]=0$. By denoting $e^{\pm \imath \varpi}$ the eigenvalues of $\Psi$ we have:
\begin{eqnarray}
  e^\pm_\mu \psi(\varphi(\theta))^\mu  & = & e^{\pm \imath \varpi} e^{\pm \imath \omega} \theta^\pm + 2\pi e^{\pm \imath \varpi} m^\pm_{i_\theta} + 2\pi \ell^\pm_{i_{\varphi(\theta)}} \\
  & = & e^{\pm \imath \omega}(e^\pm_\mu {\Psi^\mu}_\nu \theta^\nu + 2\pi \ell^\pm_{i_\theta}) \nonumber \\
  & & \quad - 2\pi e^{\pm \imath \omega} \ell^\pm_{i_\theta} + 2\pi e^{\pm \imath \varpi} m^\pm_{i_\theta} + 2\pi \ell^\pm_{i_{\varphi(\theta)}}
\end{eqnarray}
It follows that
\begin{equation}
  \psi(\varphi(\theta))^\pm - \phi^\pm_{i_{\psi \circ \varphi(\theta)}} = e^{\pm \imath \omega}(\psi(\theta)^\pm - \phi^\pm_{i_{\psi(\theta)}})
\end{equation}
if and only if
\begin{eqnarray}
  \phi^\pm_{i_{\psi \circ \varphi(\theta)}} & = & e^{\pm \imath \omega} \phi^\pm_{i_{\psi(\theta)}} - 2\pi e^{\pm \imath \omega} \ell^\pm_{i_\theta} + 2\pi e^{\pm \imath \varpi} m^\pm_{i_\theta} + 2\pi \ell^\pm_{i_{\varphi(\theta)}} \\
  \phi_{i_{\psi \circ \varphi(\theta)}} & = & \varphi(\phi_{i_{\psi(\theta)}}) - 2\pi m_{i_{\psi(\theta)}} -2\pi \Phi \ell_{i_\theta} + 2\pi \Psi m_{i_\theta} + 2\pi \ell_{i_{\varphi(\theta)}} \\
    & = & \phi_{i_{\varphi \circ \psi(\theta)}} + 2\pi(\Psi m_{i_\theta}-m_{i_{\psi(\theta)}} - \Phi \ell_{i_\theta} + \ell_{i_{\varphi(\theta)}})
\end{eqnarray}
Together with the constraint $[\Phi,\Psi]=0$, this condition is satisfied only if $\psi = \varphi^l$ for $l \in \mathbb Z$. Indeed, $[\Phi,\Psi]=[\Phi,\Phi^l]=0$, and $\varphi \circ \psi = \psi \circ \varphi = \varphi^{l+1}$. Moreover for $l\geq 0$ we have
\begin{eqnarray}
  \psi(\theta) & = & \Phi(\Phi(...(\Phi\theta + 2\pi m_{i_\theta}) +2\pi m_{i_{\varphi(\theta)}})+...)+2\pi m_{i_{\varphi^{l-1}(\theta)}} \\
  & = & \Phi^l \theta + 2\pi \underbrace{\sum_{q=0}^{l-1} \Phi^q m_{i_{\varphi^{l-1-q}(\theta)}}}_{= \ell_{i_\theta}}
\end{eqnarray}
and so
\begin{eqnarray}
  & & \Psi m_{i_\theta}-m_{i_{\psi(\theta)}} - \Phi \ell_{i_\theta} + \ell_{i_{\varphi(\theta)}} \nonumber \\
  & & \quad =\Phi^l m_{i_\theta} - m_{i_{\varphi^l(\theta)}} - \sum_{q=0}^{l-1} \Phi^{q+1} m_{i_{\varphi^{l-1-q}(\theta)}} + \sum_{q=0}^{l-1} \Phi^q m_{i_{\varphi^{l-q}(\theta)}} \\
  & & \quad =\Phi^l m_{i_\theta} - m_{i_{\varphi^l(\theta)}} - \sum_{q=1}^{l} \Phi^{q} m_{i_{\varphi^{l-q}(\theta)}} + \sum_{q=0}^{l-1} \Phi^q m_{i_{\varphi^{l-q}(\theta)}} \\
  & & \quad = 0
\end{eqnarray}
and for $l<0$ we have
\begin{eqnarray}
  \psi(\theta) & = & \Phi^{-1}(\Phi^{-1}(...(\Phi^{-1}\theta - 2\pi \Phi^{-1}m_{i_{\varphi^{-1}(\theta)}}) - 2\pi \Phi^{-1} m_{i_{\varphi^{-2}(\theta)}})\nonumber \\
  & & \qquad -...)-2\pi \Phi^{-1} m_{i_{\varphi^{-|l|}(\theta)}} \\
  & = & \Phi^{-|l|} \theta - 2\pi \underbrace{\sum_{q=1}^{|l|} \Phi^{-q} m_{i_{\varphi^{-|l|+q-1}(\theta)}}}_{= -\ell_{i_\theta}}
\end{eqnarray}
and so
\begin{eqnarray}
  & & \Psi m_{i_\theta}-m_{i_{\psi(\theta)}} - \Phi \ell_{i_\theta} + \ell_{i_{\varphi(\theta)}} \nonumber \\
  & &  =\Phi^{-|l|} m_{i_\theta} - m_{i_{\varphi^{-|l|}(\theta)}} + \sum_{q=1}^{|l|} \Phi^{-q+1} m_{i_{\varphi^{-|l|+q-1}(\theta)}} - \sum_{q=1}^{|l|} \Phi^{-q} m_{i_{\varphi^{-|l|+q}(\theta)}} \\
  & &  =\Phi^{-|l|} m_{i_\theta} - m_{i_{\varphi^{-|l|}(\theta)}} + \sum_{q=0}^{|l|-1} \Phi^{-q} m_{i_{\varphi^{-|l|+q}(\theta)}} - \sum_{q=1}^{|l|} \Phi^{-q} m_{i_{\varphi^{-|l|+q}(\theta)}} \\
  & &  = 0
\end{eqnarray}

The explicit expression of the eigenfunctions $f_\lambda$ depends on the nature real or imaginary of $\omega$.

\begin{enumerate}
\item \textit{Pure imaginary case $\omega \in \mathbb R$: }\\
$e_+ = \bar e_- \in \mathbb C^2$, and then $\theta^\pm \in \mathbb C$, moreover $\overline{\theta^+} = \theta^-$ and $\overline{\phi_i^+} = \phi_i^-$. We have then
  \begin{eqnarray}
    \arg(\varphi^{l+1}(\theta)^\pm - \phi^\pm_{i_{\varphi^{l+1}(\theta)}}) & = & \arg(\varphi^l(\theta)^\pm - \phi_{i_{\varphi^l(\theta)}}^\pm) \pm \omega \\
    |\varphi^{l+1}(\theta)^\pm - \phi^\pm_{i_{\varphi^{l+1}(\theta)}}| & = & |\varphi^l(\theta)^\pm - \phi_{i_{\varphi^l(\theta)}}^\pm|
  \end{eqnarray}
  It follows that $\forall \theta \in D_i$
  \begin{eqnarray}
    & & \mathcal K (g_\lambda(|\varphi^l(\theta)^+ - \phi_{i_{\varphi^l(\theta)}}^+|)e^{\imath  \frac{\lambda}{\omega} \arg(\varphi^l(\theta)^+ - \phi_{i_{\varphi^l(\theta)}}^+)} \nonumber \\
    & = & g_\lambda(|\varphi^{l+1}(\theta)^+ - \phi_{i_{\varphi^{l+1}(\theta)}}^+|)e^{\imath  \frac{\lambda}{\omega} \arg(\varphi^{l+1}(\theta)^+ - \phi_{i_{\varphi^{l+1}(\theta)}}^+)} \\
    & = & e^{\imath \lambda} g_\lambda(|\varphi^l(\theta)^+ - \phi_{i_{\varphi^l(\theta)}}^+|)e^{\imath  \frac{\lambda}{\omega} \arg(\varphi^l(\theta)^+ - \phi_{i_{\varphi^l(\theta)}}^+)}
  \end{eqnarray}
  for all function $g_\lambda$. $f_\lambda = g_\lambda e^{\imath \frac{\lambda}{\omega} \arg(\varphi^l(\theta)^+ - \phi_{i_{\varphi^l(\theta)}}^+)}$ are then solution of the Koopman equation. But $\arg$ is multivalued, $\arg(z)=\arg_0(z)+2k\pi$ (where $\arg_0$ is the principal value in $]-\pi,\pi[$), and then $e^{\imath  \frac{\lambda}{\omega} \arg(\theta^+ - \phi_{i_\theta}^+)}$ is ill-defined except if $\lambda \in \mathbb Z \omega$. Finally by choosing for basis of $L^2([|\phi_i^+|,|2\pi(e^+_1+e^+_2)-\phi_i^+|],dr)$ the polynomials $g_m(|\theta^+ - \phi_i^+|)=|\theta^+ - \phi_i^+|^m$, we have the following Koopman modes: $\forall n \in \mathbb Z$, $l\in\mathbb N$, $p\in \mathbb N$:
  \begin{equation}
    f_{nml}(\theta) = g_m(|\varphi^l(\theta)^+ - \phi_{i_{\varphi^l(\theta)}}^+|)e^{\imath n \arg(\varphi^l(\theta)^+ - \phi_{i_{\varphi^l(\theta)}}^+)}
  \end{equation}
  with $e^{\imath \lambda_n}= e^{\imath n \omega}$ of infinite multiplicity.

\item \textit{Real case $\omega \in \imath \mathbb R$: }\\
  $\theta^\pm, \phi_i^\pm \in \mathbb R$. We have then
  \begin{equation}
    \ln(\varphi^{l+1}(\theta)^\pm - \phi^\pm_{i_{\varphi^{l+1}(\theta)}}) = \ln(\varphi^l(\theta)^\pm - \phi_{i_{\varphi^l(\theta)}}^\pm) \pm \imath \omega
  \end{equation}
  It follows that for any $\alpha \in \mathbb R$, $\forall \theta \in D_i$
  \begin{eqnarray}
    & & \mathcal K e^{\frac{\lambda}{\omega} \left(\alpha \ln (\varphi^l(\theta)^+-\phi^+_{i_{\varphi^l(\theta)}}) - (1-\alpha) \ln (\varphi^l(\theta)^--\phi^-_{i_{\varphi^l(\theta)}})\right)} \nonumber \\
    & = & e^{\frac{\lambda}{\omega} \left(\alpha \ln (\varphi^{l+1}(\theta)^+-\phi^+_{i_{\varphi^{l+1}(\theta)}}) - (1-\alpha) \ln (\varphi^{l+1}(\theta)^--\phi^-_{i_{\varphi^{l+1}(\theta)}})\right)} \\
  & = &  e^{\frac{\lambda}{\omega} \left(\alpha (\ln (\varphi^l(\theta)^+-\phi^+_{i_{\varphi^l(\theta)}})+\imath \omega) - (1-\alpha) (\ln (\varphi^l(\theta)^--\phi^-_{i_{\varphi^l(\theta)}})-\imath\omega)\right)} \\
  & = & e^{\imath \lambda} e^{\frac{\lambda}{\omega} \left(\alpha \ln (\varphi^l(\theta)^+-\phi^+_{i_{\varphi^l(\theta)}}) - (1-\alpha) \ln (\varphi^l(\theta)^--\phi^-_{i_{\varphi^l(\theta)}})\right)}
\end{eqnarray}
  The solutions to the Koopman equation are then
  \begin{equation}
    f_{\lambda\alpha l}(\theta)  =  e^{-\imath \frac{\lambda}{\IM \omega} \left(\alpha \ln (\varphi^l(\theta)^+-\phi^+_{i_{\varphi^l(\theta)}}) - (1-\alpha) \ln (\varphi^l(\theta)^--\phi^-_{i_{\varphi^l(\theta)}})\right)}
  \end{equation}
\end{enumerate}
\end{proof}

These Koopman modes reflect the fact that the neighbourhood of the orbit of a cyclic point $\phi^\pm_{i_\theta}$ is ``driven'' by its dynamics, inducing the existence of observables (these Koopman modes) stable under this cyclicality. We discuss in more detail the relation between these Koopman modes and the properties of $\varphi$ in section III, by treating the different cases which occur. The global behaviour of the flow reflected by these full Koopman modes are the cyclic drivings of the phase space points and consequently of the phase space observables. The full Koopman modes are then continuous on the phase space regions influenced by a single driving (as we will see it section III), ensuring the coherence on the whole of $\Gamma$. Moreover if $\omega \in \mathbb R$, they constitute countable basis for the eigenspaces (even if $\overline{\{e^{\imath n \omega}\}_{n \in \mathbb Z}}$ is continuous, i.e. $\omega \not\in 2\pi \mathbb Q$).\\

The normalisation factor of $e^\pm$ being arbitrary (since $(e^\pm)$ are bi-orthonormal to $(e_\pm)$), the expression of the complex points $\phi^\pm_{i_\theta}$ depends on this arbitrary choice (the complexification of the torus $\Gamma_\theta$ depends on this arbitrary choice). Under a gauge change $\tilde e^\pm = \pm \beta e^\pm$ (with $\beta \in \mathbb C$), we have in the real case $\tilde f_{nml} = |\beta|^m e^{\imath n \arg \beta} f_{nml}$ and in the pure imaginary case $\tilde f_{\lambda \alpha l} = e^{\imath \frac{\lambda}{\IM(\omega)} \beta} f_{\lambda \alpha l}$ (these are just changes of normalisation and phase factors of the Koopman modes).\\

In the case $\omega \in \mathbb R$, we can note that if $\frac{\omega}{2\pi} \not\in \mathbb Q$, then $\Sp(\mathcal K) = \overline{\{e^{\imath n \omega}\}_{n \in \mathbb Z}} = U(1)$. In contrast, in the resonant case $\frac{\omega}{2\pi} \in \mathbb Q$, we have $\omega = \frac{2\pi q}{N}$ with $q,N\in \mathbb N$ coprimes and $\Sp(\mathcal K) = \{e^{\imath \frac{2\pi}{N}n}\}_{n \in \{0,...,N-1\}}$. We have then $\Phi^N = \id$. It follows that $\varphi^N(\theta)^\pm-\phi^\pm_{i_{\varphi^N(\theta)}} = \theta^\pm - \phi^\pm_{i_{\theta}}$. The degeneracy of $\lambda_n$ is then restricted to $l \in \{0,...,N-1\}$. \\
If moreover the patches partition of $\Gamma$ is invariant, i.e. $\varphi(D_i) = D_i$ ($\forall i$), then $\varphi(\theta)^\pm - \phi^\pm_{i_{\varphi(\theta)}} = e^{\pm \imath \omega} \theta^\pm + 2\pi e_\mu^\pm \phi^\mu_{i_\theta} - \phi^\pm_{i_\theta} = e^{\pm \imath \omega}(\theta^\pm - \phi^\pm_{i_\theta})$ (since $i_{\varphi(\theta)}=i_\theta$). In a same way, by induction, $\varphi^l(\theta)^\pm - \phi^\pm_{i_{\varphi^l(\theta)}} = e^{\pm \imath l \omega}(\theta^\pm-\phi^\pm_{i_\theta})$ (indeed $\varphi^l(\theta)^\pm - \phi^\pm_{i_{\varphi^l(\theta)}} = e^{\pm \imath l \omega}(\theta^\pm-\phi^\pm_{i_\theta}) \Rightarrow \varphi^{l+1}(\theta)^\pm - \phi^\pm_{i_{\varphi^{l+1}(\theta)}} = e^{\pm \imath \omega} \varphi^l(\theta)^\pm + 2\pi e_\mu^\pm m^\mu_{i_\theta} - \phi^\pm_{i_\theta} = e^{\pm \imath \omega} (e^{\pm \imath l \omega}(\theta^\pm-\phi^\pm_{i_\theta}) + \phi^\pm_{i_\theta}) -e^{\pm \imath \omega} \phi^\pm_{i_\theta} = e^{\pm \imath (l+1)\omega} (\theta^\pm-\phi^\pm_{i_\theta})$). It follows that $\forall l$, $f_{nml} = e^{\imath l \lambda_n} f_{nm0}$. $\lambda_n$ is then not degenerate with $l$ which can be fixed to $0$.\\

We can now examine in detail the relation between the Koopman modes of ergodic components and the full Koopman modes, with respect to the dynamical behaviour of the cat map.

\section{The Koopman modes with respect to the different classical dynamical regimes}
\subsection{$N$-cyclic cat map case}
We consider the case $\Phi \in SL(2,\mathbb Z)$ where $\omega = \frac{2\pi q}{N}$ with $q,N$ coprime, and then where $\Phi^N=\id$. In this case $\varphi^N(\theta) = \Phi^N \theta + 2\pi \sum_{k=0}^{N-1} \Phi^k m_{i_{\varphi^{N-1-k}(\theta)}} = \theta \mod 2\pi = \theta$; all orbits are at least $N$-cyclic. $\Omega$ is 2 dimensional and we have $\mu(\Omega \cap \varphi^n(\Omega))=0$, $\forall n \in \{1,...,N-1\}$ and $\overline{\bigcup_{n=0}^{N-1} \varphi^n(\Omega)} = \Gamma$. $\forall \theta_0 \in \Omega$, the Koopman modes of the ergodic component $\Gamma_{\theta_0} = \{\varphi^n(\theta_0)\}_{n\in\{0,...,N-1\}}$ are $\forall n \in \{0,...,N-1\}$
\begin{equation}
  f_{n\theta_0}(\theta) = \sum_{p=0}^{N-1} e^{\imath \frac{2\pi}{N} np} \delta_{\varphi^p(\theta_0)}(\theta) \text{ with } e^{\imath \lambda_n} = e^{\imath \frac{2\pi}{N}n}
\end{equation}
with $f_{n\theta_0} \in \mathbb C^N_{\theta_0}$ (the finite dimensional Hilbert space generated by $(\delta_{\varphi^p(\theta_0)})_{p \in \{0,...,N-1\}}$).
\begin{equation}
  L^2(\Gamma,d\mu) = \int_\Omega^\oplus \mathbb C^N_{\theta_0} \frac{d\mu(\theta_0)}{\mu(\Omega)}
\end{equation}

Let $f_{nml}$ ($l < N$) be a full Koopman mode eq. (\ref{solcat}) associated with $e^{\imath \frac{2\pi q}{N} n}$. Due to the $N$-cyclicality and the covering by the ergodic components we can write:
\begin{eqnarray}
  f_{nml} & = & \int_{\Omega}^\oplus \sum_{p=0}^{N-1} f_{nml}(\varphi^p(\theta_0)) \delta_{\varphi^p(\theta_0)} \frac{d\mu(\theta_0)}{\mu(\Omega)} \\
  & = & \int_\Omega^\oplus \sum_{p=0}^{N-1} |\varphi^l(\theta_0)^+-\phi_{i_{\varphi^l(\theta_0)}}^+|^m e^{\imath pn \frac{2\pi q}{N}} e^{\imath n\arg(\varphi^l(\theta_0)^+-\phi^+_{i_{\varphi^l(\theta_0)}})} \nonumber \\
  & & \quad \times \delta_{\varphi^p(\theta_0)} \frac{d\mu(\theta_0)}{\mu(\Omega)} \\
  & = & \int_\Omega^\oplus |\varphi^l(\theta_0)^+-\phi_{i_{\varphi^l(\theta_0)}}^+|^m e^{\imath pn \frac{2\pi q}{N}} e^{\imath n\arg(\varphi^l(\theta_0)^+-\phi^+_{i_{\varphi^l(\theta_0)}})} \nonumber \\
  & & \quad \times \sum_{p=0}^{N-1} e^{\imath p (nq \mod N)\frac{2\pi}{N}} \delta_{\varphi^p(\theta_0)} \frac{d\mu(\theta_0)}{\mu(\Omega)} \\
  & = & \int_\Omega^\oplus |\varphi^l(\theta_0)^+-\phi_{i_{\varphi^l(\theta_0)}}^+|^m e^{\imath n \arg(\varphi^l(\theta_0)^+-\phi^+_{i_{\varphi^l(\theta_0)}})} \nonumber \\
  & & \quad \times f_{(nq \mod N)\theta_0}  \frac{d\mu(\theta_0)}{\mu(\Omega)}
\end{eqnarray}

The Koopman spectrum is pure point: $\Sp(\mathcal K) = \Sp_{pp}(\mathcal K) = \{e^{\frac{2\imath \pi}{N} n}\}_{n \in \{0,...,N-1\}}$, but not discrete since each eigenvalue is of infinite multiplicity:
\begin{eqnarray}
  & & \ker(\mathcal K - e^{\frac{2\imath \pi}{N}n}) \nonumber \\
  & & = \Lin(f_{n\theta_a})_{\scriptstyle \theta_a \in \Omega}\\
  & & = \Lin\left(|\varphi^l(\theta)^+-\phi_{i_{\varphi^l(\theta)}}^+|^m e^{\imath k_n \arg (\varphi^l(\theta)^+-\phi^+_{i_{\varphi^l(\theta)}})}\right)_{\begin{array}{l} \scriptstyle m \in \mathbb N \\ \scriptstyle l \in\{0,...,N-1\} \\ \scriptstyle k_n \text{ s.t. } k_nq \mod N = n \end{array}}
\end{eqnarray}

In the generic case, the first return time in $D_i$ is $N$ ($\forall i$), and $\sum_{k=0}^{N-1} \Phi^k m_{i_{\varphi^{N-1-k}(\theta)}}=0$ induces that the formula of $\phi_{i_\theta}$ does not hold. The full Koopman modes are then reduced to
\begin{equation}
  f_{nml}(\theta) = |\varphi^l(\theta)^+|^m e^{\imath n \arg \varphi^l(\theta)^+}
\end{equation}
for $ n\in \mathbb Z, m \in \mathbb N, l \in \{0,...,N-1\}$.\\

 More interesting situations occur when the first return time in $D_i$ is lower than $N$, for example when the partition $\{D_i\}_i$ is invariant by $\varphi$.

  \begin{example}
  Let $\Phi = \left(\begin{array}{cc} -1 & 1 \\ -1 & 0 \end{array}\right)$ be for which $\Phi^3 = \id$. Almost all $\theta \in \Gamma$ are then $3$-cyclic. The two patches are $D_1 = \{(\theta^1,\theta^2)\in [0,2\pi[^2, \theta^2 \leq \theta^1 \}$ and $D_2= \{(\theta^1,\theta^2)\in [0,2\pi[^2, \theta^2 \geq \theta^1 \}$, with the associated translations $m_1 = (1,1)$ and $m_2=(0,1)$, see fig. \ref{cat1patches}.
          \begin{figure}
            \begin{center}
            \includegraphics[width=5cm]{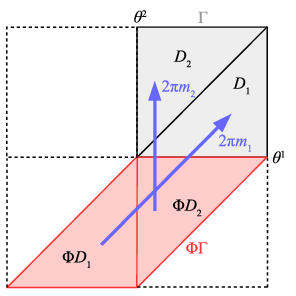}\\
            \caption{\label{cat1patches} Graphical representation of $\varphi: \theta \mapsto \Phi \theta + 2\pi m_i$, with $\Phi = \left(\begin{array}{cc} -1 & 1 \\ -1 & 0 \end{array}\right)$. $\Gamma$ is represented by $[0,2\pi[^2$.}
            \end{center}
          \end{figure}
The partition is here invariant $\varphi(D_1)=D_1$ and $\varphi(D_2)=D_2$. The time of first return into a patch is then 1 for all $\theta$ (we are then interested by the complex fixed points of the system) and the degeneracy of the Koopman value with the iterations of $\varphi$ is killed.\\          

 $\Sp(\Phi) = \{e^{\pm \imath \frac{2\pi}{3}}\}$ ($\omega = \frac{2\pi}{3}$), with the following right and left eigenvectors:
    \begin{equation}
      e_\pm = \frac{1}{e^{\mp \imath \frac{2\pi}{3}}-1} \left(\begin{array}{c} e^{\pm \imath \frac{2\pi}{3}} \\ -1 \end{array} \right) \qquad e^\pm = \left(\begin{array}{cc} e^{\pm \imath \frac{2\pi}{3}} & 1 \end{array} \right)
    \end{equation}
    The complex coordinates of the complexification $\Gamma_{\mathbb C}$ are then $\theta^\pm  =  e^{\pm \imath \frac{2\pi}{3}} \theta^1+\theta^2$. The complex fixed points associated with the two patches are
    \begin{eqnarray}
      \phi^\pm_1 & = & \frac{2\pi(e^{\pm \imath \frac{2\pi}{3}}+1)}{1-e^{\pm \imath \frac{2\pi}{3}}} = \pm 2\pi \imath \frac{\sqrt 3}{3} \\
      \phi^\pm_2 & = & \frac{2\pi}{1-e^{\pm \imath \frac{2\pi}{3}}} = \pi \pm \imath \pi \frac{\sqrt 3}{3}
    \end{eqnarray}
    These two fixed points are in real coordinates in $\Gamma$: $\phi_1 = (\frac{4\pi}{3},\frac{2\pi}{3})$ and $\phi_2=(\frac{2\pi}{3},\frac{4\pi}{3})$. Indeed in complex coordinates
    \begin{eqnarray}
      e^\pm_\mu \phi_1^\mu & = & \pm 2\pi \imath \frac{\sqrt 3}{3} = \phi^\pm_1 \\
      e^\pm_\mu \phi_2^\mu & = & \pi \pm \imath \pi \frac{\sqrt 3}{3} = \phi^\pm_2
    \end{eqnarray}
    We find $\Omega$ by connecting the fixed points -- including $0$ -- (which are at the intersections of $\Omega$ and its images), for example we can choose (see fig. \ref{cat1Omega}):
    \begin{equation}
      \Omega = \{(\theta^1,\theta^2)\in[0,2\pi[^2, \max(\frac{\theta^1}{2},2\theta^1-2\pi) \leq \theta^2 \leq \min(\frac{\theta^1}{2}+\pi,2\theta^1) \}
    \end{equation}

    \begin{figure}
      \begin{center}
        \includegraphics[width=5cm]{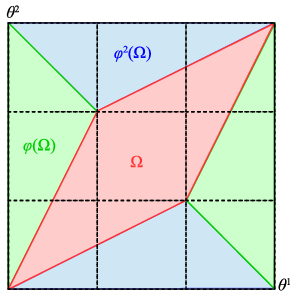}\includegraphics[width=5cm]{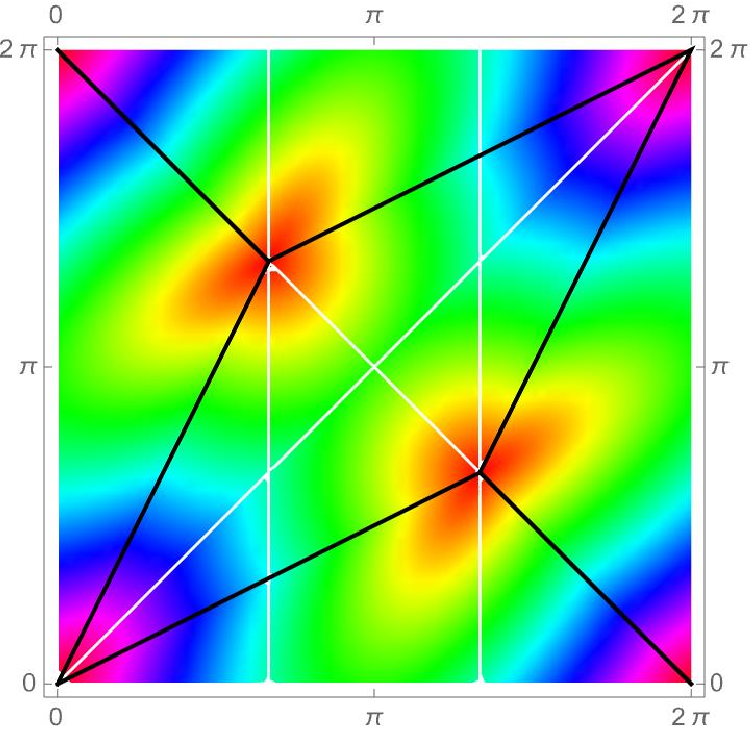}\\
        \caption{\label{cat1Omega} $\Omega$, $\varphi(\Omega)$ and $\varphi^2(\Omega)$ for $\Phi = \left(\begin{array}{cc} -1 & 1 \\ -1 & 0 \end{array}\right)$. Right: same as left but with a colouring of the orbits according to the distance from their initial points to the nontrivial fixed points.}
      \end{center}
    \end{figure}

    $\Sp(\mathcal K) = \{1,e^{\imath \frac{2\pi}{3}},e^{\imath \frac{4\pi}{3}}\}$ and the Full Koopman modes associated with $e^{\imath n \frac{2\pi}{3}}$ are then $\forall m \in \mathbb N$ 
    \begin{equation} \label{modescat1}
      f_{nm}(\theta) = \left\{ \begin{array}{cc} |e^{\imath \frac{2\pi}{3}} \theta^1 + \theta^2 -\phi_1^+|^m e^{\imath n \arg( e^{\imath \frac{2\pi}{3}} \theta^1 + \theta^2 -\phi_1^+)} & \text{if } \theta \in D_1 \\
      |e^{\imath \frac{2\pi}{3}} \theta^1 + \theta^2 -\phi_2^+|^m e^{\imath n \arg( e^{\imath \frac{2\pi}{3}} \theta^1 + \theta^2 -\phi_2^+)} & \text{if } \theta \in D_2
      \end{array} \right.
    \end{equation}
    These ones are represented fig. \ref{KoopModescat1}.
      \begin{figure}
        \begin{center}
          \includegraphics[width=10cm]{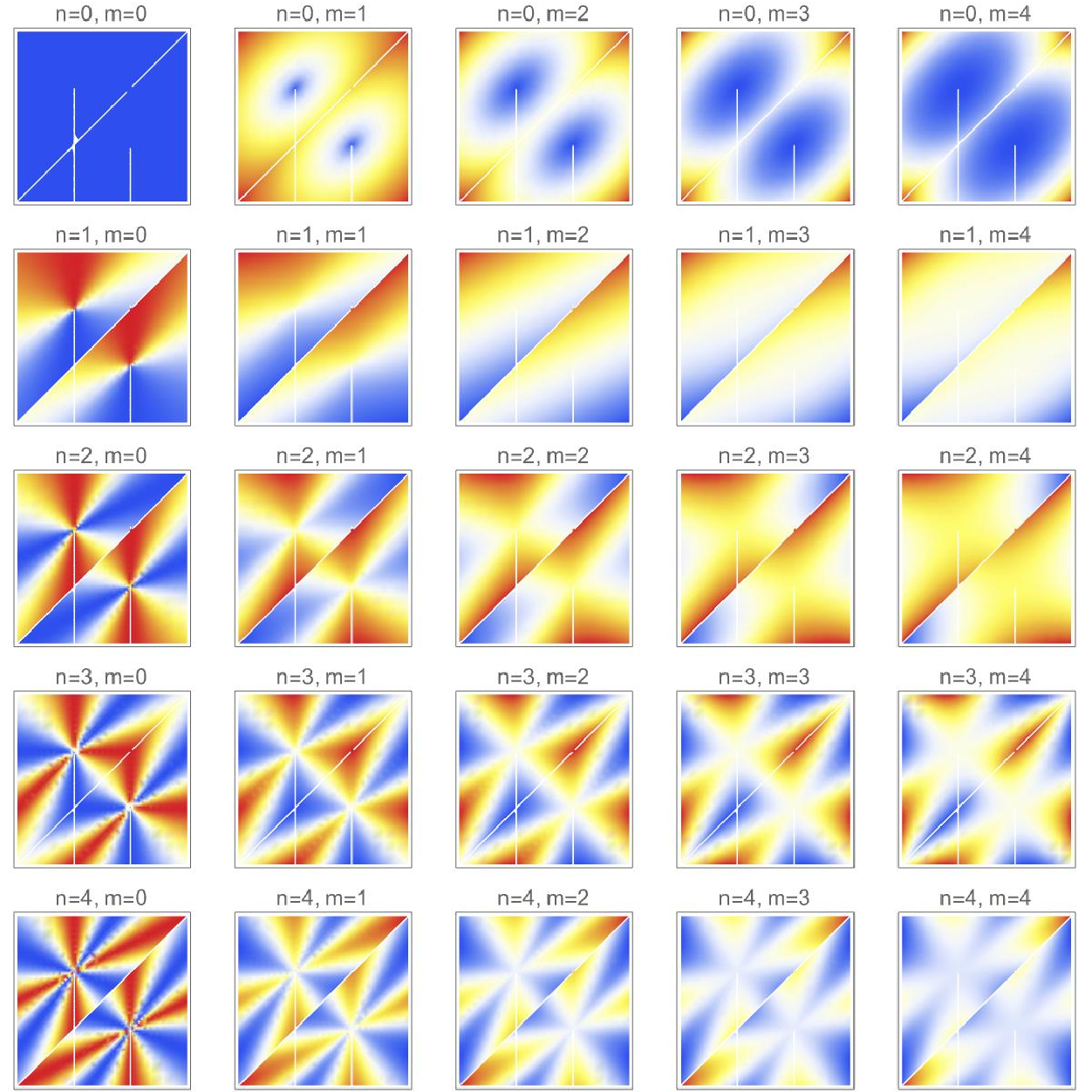}\\
          \caption{\label{KoopModescat1} The real part of the Koopman modes $\RE(f_{nm}(\theta))$ eq. (\ref{modescat1}) for different values of $n$ and $m$, for the cat system defined by $\Phi = \left(\begin{array}{cc} -1 & 1 \\ -1 & 0 \end{array}\right)$ for which $\frac{\omega}{2\pi}=\frac{1}{3}$.}
        \end{center}
      \end{figure}
      The eigensubspaces are $\forall n \in \{0,1,2\}$, $\ker(\mathcal K-e^{\imath \frac{2\pi n}{3}}) = \{f_{n+3p,m}\}_{p \in \mathbb Z, m \in \mathbb N}$. These Koopman modes are characteristic of the fact that the orbit of a point in $D_i$ turns around $\phi_i$ inside $D_i$ (fig. \ref{cat1Omega}). The dynamics inside $D_1$ and $D_2$ are then completely independent (the partition being invariant) and are essentially rotation around the non-trivial fixed points. The stable observables coherent on the whole of $\Gamma$, the full Koopman modes (fig. \ref{KoopModescat1}) reflect this behaviour, with inside each patch a central symmetry with respect to the fixed points. By construction, $n$ is the ``wave number'' of the oscillations of the stable observables around the nontrivial fixed points, whereas $m$ is the power of the amplitude growing of the stable observables (waves) with respect to the distance to the nontrivial fixed points. Noting than $m$ being a degeneracy index, we can built by superposition full Koopman modes of the form $g(|e^{\imath \frac{2\pi}{3}} \theta^1 + \theta^2 +\phi_1^+|)e^{\imath n \arg( e^{\imath \frac{2\pi}{3}} \theta^1 + \theta^2 +\phi_1^+)}$ for any analytical function $g$.
\end{example}

\subsection{Quasi-cyclic cat map case}
We consider the case where $\omega \not\in 2\pi \mathbb Q$. The orbits are quasi-cyclic: let $\epsilon>0$ and $q_\epsilon,N_\epsilon \in \mathbb N$ be coprime integers such that $|\frac{\omega}{2\pi} - \frac{q_\epsilon}{N_\epsilon}|\leq \epsilon$. Let $\eta_\epsilon = \frac{\omega}{2\pi} - \frac{q_\epsilon}{N_\epsilon}$. We have then the quasi-cyclicality:
\begin{eqnarray}
  \Phi^{N_\epsilon} & = & \left(\begin{array}{cc} e^{\imath \omega N_\epsilon} & \\ 0 & e^{-\imath \omega N_\epsilon} \end{array} \right)_{(e_a)} \\
  & = & \left(\begin{array}{cc} e^{\imath 2\pi q_\epsilon + \imath 2\pi N_\epsilon \eta_\epsilon} & \\ 0 & e^{-\imath 2\pi q_\epsilon - \imath 2\pi N_\epsilon \eta_\epsilon } \end{array} \right)_{(e_a)} \\
  & = & \id+ 2\imath \pi N_\epsilon \eta_\epsilon \left(\begin{array}{cc} 1 &0 \\ 0 & -1 \end{array}\right)_{(e_a)} + \mathcal O(\epsilon^2)
\end{eqnarray}

$\Omega$ is 1 dimensional, let $a \in [0,\ell(\Omega)[$ (where $\ell(\Omega)$ is the length of the curve $\Omega$) be a curvilinear coordinates onto $\Omega$ ($\theta_a \in \Omega \subset \Gamma$ is a point parametrised by $a$). The ergodic component $\Gamma_a = \overline{\{\varphi^n(\theta_a)\}_{n \in \mathbb Z}}$ is a closed path in $\Gamma$, not necessarily connected ($\Gamma_a$ can be a set of several disjoint connected closed paths - cycles -). Let $n_a$ be the number of connected cycles forming $\Gamma_a$, we have
    \begin{equation}
      L^2(\Gamma,d\mu) = \int_\Omega^\oplus \bigoplus_{i=1}^{n_a} L^2(\Gamma_a^i,\frac{d\vartheta^i}{2\pi}) \frac{da}{\ell(\Omega)}
    \end{equation}
    where $\Gamma_a^i$ is a connected cycle of $\Gamma_a$ and $\vartheta^i \in [0,2\pi[$ is its local coordinates. The Koopman modes of an ergodic component are $\forall e^{\imath \lambda} \in U(1)$:
    \begin{equation}
      f_{\lambda a} = \sum_{p=-\infty}^{+\infty} e^{\imath p \lambda} \delta_{\varphi^p(\theta_a)}
    \end{equation}
    For $\lambda = n\omega$ ($n \in \mathbb Z$), we have
    \begin{eqnarray}
  f_{n\omega, a} & = & \sum_{p=-\infty}^{+\infty} e^{\imath pn \omega} \delta_{\varphi^p(\theta_a)} \\
  & = & \sum_{p=-\infty}^\infty e^{\imath pn 2\pi \eta_\epsilon} e^{\imath pn \frac{2\pi q_\epsilon}{N_\epsilon}} \delta_{\varphi^p(\theta_a)} \\
  & = & \sum_{k=-\infty}^{+\infty} \sum_{p=0}^{N_\epsilon-1} e^{\imath (p+kN_\epsilon)n2\pi \eta_\epsilon} e^{\imath pn \frac{2\pi q_\epsilon}{N_\epsilon}} \delta_{\varphi^{p+kN_\epsilon}(\theta_a)}
    \end{eqnarray}
    in which we find an approximated Koopman mode of a $N_\epsilon$-cyclic system $\sum_{p=0}^{N_\epsilon-1} e^{\imath p(nq_\epsilon \mod N_\epsilon) \frac{2\pi}{N_\epsilon}} \delta_{\varphi^{p}(\theta_0)} + \mathcal O(\epsilon)$. For $\lambda \not\in \mathbb Z\omega$, we can find $n_\epsilon$ such that $\lambda = n_\epsilon \omega + \mathcal O(\epsilon^2)$.\\

    Let $f_{nml}$ be a full Koopman mode eq. (\ref{solcat}) associated with $e^{\imath n \omega}$. Due to the covering by the ergodic components we can write:
    \begin{eqnarray}
  f_{nml} & = & \int_{\Omega}^\oplus \sum_{p=-\infty}^{+\infty} f_{nml}(\varphi^p(\theta_a)) \delta_{\varphi^p(\theta_a)} \frac{da}{\ell(\Omega)} \\
  & = & \int_{\Omega}^\oplus |\varphi^l(\theta_a)^+-\phi_{i_{\varphi^l(\theta_a)}}^+|^m e^{\imath n\arg(\varphi^l(\theta_a)^+-\phi_{i_{\varphi^l(\theta_a)}}^+)} \nonumber \\
  & & \qquad \times \sum_{p=-\infty}^{+\infty} e^{\imath pn \omega} \delta_{\varphi^p(\theta_a)} \frac{da}{\ell(\Omega)} \\
  & = & \int_{\Omega}^\oplus |\varphi^l(\theta_a)^+-\phi_{i_{\varphi^l(\theta_a)}}^+|^m e^{\imath n\arg(\varphi^l(\theta_a)^+-\phi_{i_{\varphi^l(\theta_a)}}^+)} f_{n\omega, a}  \frac{da}{\ell(\Omega)}
    \end{eqnarray}

    $\forall e^{\imath \lambda} \in U(1) \setminus \{e^{\imath n \omega}\}_{n \in \mathbb Z}$, $\exists (n_p)_p \in \mathbb N^{\mathbb N}$ such that $\lim_{p \to +\infty} e^{\imath n_p \omega} = e^{\imath \lambda}$ (i.e. $\lim_{p \to +\infty} n_p \omega = \lambda \mod 2\pi$). For a fixed value of $m \in \mathbb N$, $(f_{n_p ml})_{p \in \mathbb N}$ is a Weyl sequence for $e^{\imath \lambda}$. $(f_{n_p ml})_{p \in \mathbb N}$  weakly converges to $f_{\lambda m l}(\theta) = |\varphi^l(\theta)^+-\phi_{i_{\varphi^l(\theta)}}^+|^m e^{\imath \frac{\lambda}{\omega} \arg(\varphi^l(\theta)^+-\phi^+_{i_{\varphi^l(\theta)}})}$.\\

    $\Sp(\mathcal K) = U(1)$. $e^{\imath n \omega}$ is a synthetic spectral value of the continuum which is also an approximate synthetic eigenvalue of $\left. \mathcal K \right|_{\Gamma_a}$ associated with the approximate eigenvector $\sum_{q=0}^{N_\epsilon-1} e^{\imath p n\omega} \delta_{\varphi^p(\theta_a)}$. But $e^{\imath n \omega}$ is a true eigenvalue of $\mathcal K$ (because of $f_{nml} \in L^2(\Gamma,d\mu)$). We have then the case of synthetic spectral values of the continuum of $\left. \mathcal K \right|_{\Gamma_a}$  which become pure point eigenvalues of $\mathcal K$ (a situation related here to the regular quasi-cyclicality of the dynamics). We have then $\Sp_{pp}(\mathcal K) = \{e^{\imath n \omega}\}_{n \in \mathbb N}$. In other words, since $f_{n\omega,a}$ is in ``synthesis resonance'' with all $f_{n\omega,a'}$ (with $a'$ in the neighbourhood of $a$ in $\Omega$), the synthetic value of the continuum $n\omega$ becomes a pure point eigenvalue of the synthetised operator $\mathcal K$.\\
    Since the cat map presents also some cyclic orbits of any period $N$ ($\{\phi_i^\pm\}$), these ones are special ergodic components. We have then $\Sp(\left. \mathcal K \right|_{\Orb(\theta_*)}) = \{e^{\imath \frac{2\pi n}{N_{\theta_*}}}\}_{n \in \{0,...,N_{\theta_*}-1\}}$ for $\theta_*$ a $N_{\theta_*}$-cyclic point. But since the set of these cyclic points is of zero measure, these eigenvalues are not synthetic.

\begin{example}
  Let $\Phi = \left(\begin{array}{cc} 2\cos \omega & 1 \\ -1 & 0 \end{array} \right)$ with  $\omega = \sqrt 2 \pi$ (and then $\frac{\omega}{2\pi} = \frac{\sqrt 2}{2} \not\in \mathbb Q$). The two patches are $D^1=\{(\theta^1,\theta^2) \in [0,2\pi[^2, \theta^2 \leq -2\cos(\omega) \theta^1\}$ and $D^2=\{(\theta^1,\theta^2) \in [0,2\pi[^2, \theta^2 \geq -2\cos(\omega) \theta^1\}$ with the associated translations $m_1=(1,1)$ and $m_2=(0,1)$ see fig. \ref{cat2patches}.
          \begin{figure}
            \begin{center}
              \includegraphics[width=5cm]{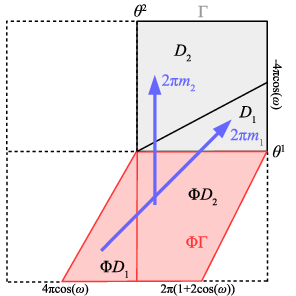} \includegraphics[width=5.5cm]{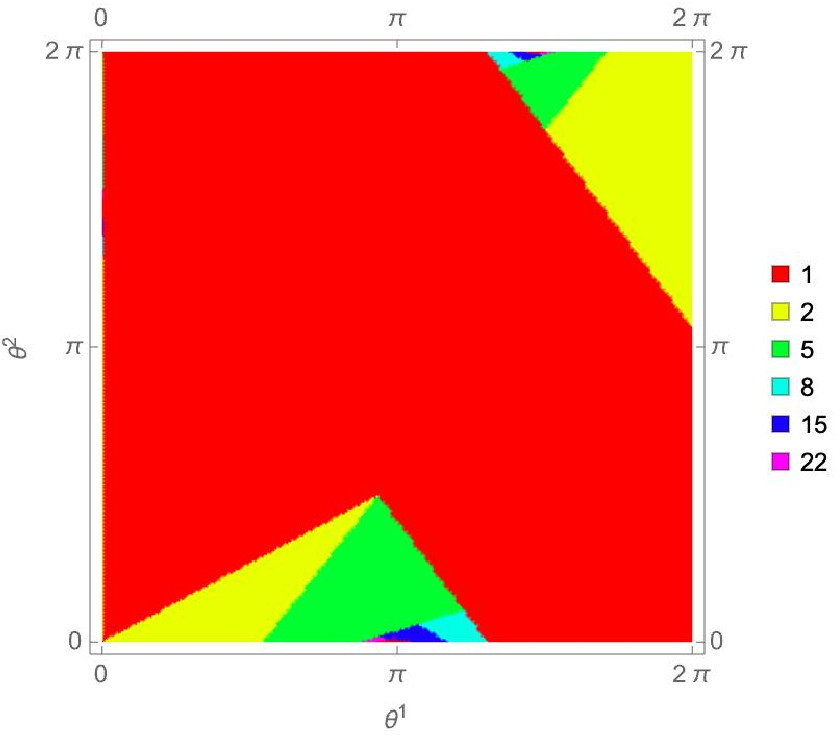}\\
              \caption{\label{cat2patches} Graphical representation of $\varphi: \theta \mapsto \Phi \theta + 2\pi m_i$ (left), with $\Phi = \left(\begin{array}{cc} 2 \cos \omega & 1 \\ -1 & 0 \end{array}\right)$; and for each point $\theta$, the value $N_\theta$ of its first return time in its initial patch $D_{i_\theta}$ (right). $\Gamma$ is represented by $[0,2\pi[^2$.}
            \end{center}
          \end{figure}
The partition being not invariant nor cyclic, $N_\theta$ the first return time in $D_{i_\theta}$ of the orbit of $\theta$ changes with $\theta$.\\

$\Sp(\Phi) = \{e^{\pm \imath \sqrt 2 \pi}\}$ with the following right and left eigenvectors:
    \begin{equation}
      e_\pm = \frac{1}{e^{\pm 2 \imath \omega}-1} \left(\begin{array}{c} e^{\pm \imath \omega} \\ -1 \end{array} \right) \qquad e^\pm = \left(\begin{array}{cc} e^{\pm \imath \omega} & 1 \end{array} \right)
    \end{equation}
    The coordinates of the complexification of the torus $\Gamma_{\mathbb C}$ are then $\theta^\pm  =  e^{\pm \imath \omega} \theta^1+\theta^2$. For the points $\theta$ returning immediately in their initial patch (red in the right of fig. \ref{cat2patches}), the associated complex fixed points are
    \begin{eqnarray}
      \phi^\pm_1 & = & 2\pi \frac{e^{\pm \imath \omega} + 1}{1-e^{\pm \imath \omega}} = \pm \imath \frac{2\pi \sin \omega}{1-\cos \omega}\\
      \phi^\pm_2 & = & \frac{2\pi}{1-e^{\pm \imath \omega}} = \pi \pm \imath \pi \frac{\sin \omega}{1-\cos \omega}
    \end{eqnarray}
    In real coordinates ($\phi^\pm_a = e^\pm_\mu \phi^\mu_a$) we have
    \begin{eqnarray}
      \phi_1 & = & \left(\frac{2\pi}{1-\cos \omega} , - \frac{2\pi \cos \omega}{1-\cos \omega} \right)\\
      \phi_2 & = & \left(\frac{\pi}{1-\cos \omega} , \pi \frac{1-2\cos \omega}{1-\cos \omega} \right)
    \end{eqnarray}

    For the other cases where the first return time $N_\theta$ is not 1, we have more generally:
    \begin{eqnarray}
      \phi_{i_\theta}^\pm & = & \frac{2\pi \sum_{q=0}^{N_{\theta}-1} e^{\pm \imath (N_\theta-1-q)\omega} (e^{\pm \imath \omega} \pmb 1_{D_1}(\varphi^q(\theta))+ 1)}{1-e^{\pm \imath N_\theta \omega}} \\
      & = & 2\pi \frac{\sum_{q=0}^{N_\theta-1} e^{\pm \imath(N_\theta-q)\omega} \pmb 1_{D_1}(\varphi^q(\theta))}{1-e^{\pm \imath N_\theta \omega}} + \frac{2\pi}{1-e^{\pm \imath \omega}} \\
      & = & \left\{\begin{array}{cc} \frac{2 \pi e^{\pm \imath N_\theta \omega}}{1-e^{\pm \imath N_\theta \omega}} + \frac{2\pi}{1-e^{\pm \imath \omega}} & \text{if } i_\theta=1 \\
      \frac{4\pi}{1-e^{\pm \imath \omega}} - \frac{2\pi}{1-e^{\pm \imath N_\theta \omega}} & \text{if } i_\theta = 2 \end{array} \right.
    \end{eqnarray}
    where $N_\theta$ is the time of first return of the orbit of $\theta$ into $D_{i_\theta}$ and $\pmb 1_{D_1}$ is the indicator function of $D_1$. We have used the fact that if $i_\theta =1$ then $i_{\varphi(\theta)}=...=i_{\varphi^{N_\theta-1}(\theta)}=2$ (the first return to $D_1$ being after $N_\theta$ iterations); and if $i_\theta=2$ then $i_{\varphi(\theta)}=...=i_{\varphi^{N_\theta-1}(\theta)}=1$. $\pmb 1_{D_1}(\varphi^q(\theta))=0$ if $i_\theta=1$ except for $q=0$ and $\pmb 1_{D_1}(\varphi^q(\theta))=1$ if $i_\theta=2$ except for $q=0$.\\
    
      The ergodic components $\{\Gamma_a\}_{a \in \Omega}$ are organised around the cyclic points of $\varphi$, see fig. \ref{cat2Cycles}.
    \begin{figure}
      \begin{center}
        \includegraphics[width=5cm]{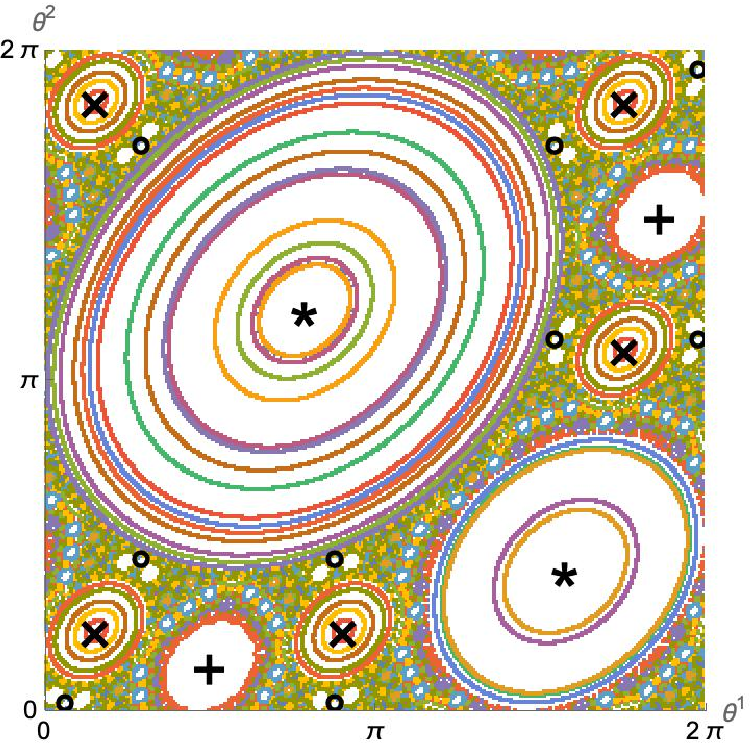} \includegraphics[width=5cm]{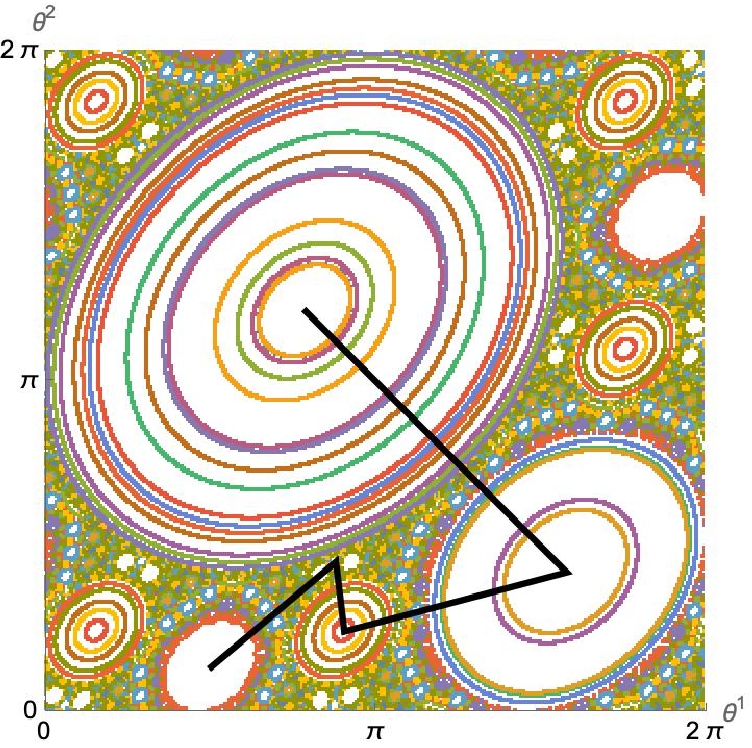}\\
        \caption{\label{cat2Cycles} Representation of some orbits $\{\varphi^n(\theta_0)\}_{n \in \mathbb N}$ in $\Gamma$ (with $\Phi = \left(\begin{array}{cc} 2\cos\omega & 1 \\ -1 & 0 \end{array}\right)$), with the position of the fixed points ($*$), 2-cyclic points ($+$), 5-cyclic points ($\times$), and 8-cyclic points ($\circ$) associated with the first return in $D_{i_\theta}$. Right: rough approximation of $\Omega$.}
      \end{center}
    \end{figure}
    Around the non-trivial fixed point (the trivial fixed point being $(0,0)$) $\Gamma_{a}$ are single cycles, but around the $p$-cyclic points $\Gamma_{a}$ are the disjoint union of $p$ cycles (the orbit of $\theta_a$ jump from a cycles to another one as the centre runs between the $p$ points of the $p$-cyclic orbit). $\Omega$ must be a path crossing each disconnected components $\Gamma_{a}$ once and only once. For which $\Omega$ can be choose as a path linking one point of each $p$-cyclic orbit of $\varphi$, $\forall p$. A rough approximation of $\Omega$ taking into account only the cyclic points of law period is represented fig. \ref{cat2Cycles}.\\
    The other cyclic points (not associated with the first return times) are just exceptional cyclic orbits immersed in the islands of cycles of quasi-cyclic orbits in the fig. \ref{cat2Cycles}.\\

    $\Sp(\mathcal K) = U(1)$ with $\Sp_{pp}(\mathcal K) = \{e^{\imath n \omega}\}_{n \in \mathbb Z}$, the full Koopman modes associated with $e^{\imath n \omega}$ are represented fig. \ref{KoopModescat2}
          \begin{figure}
        \begin{center}
          \includegraphics[width=10cm]{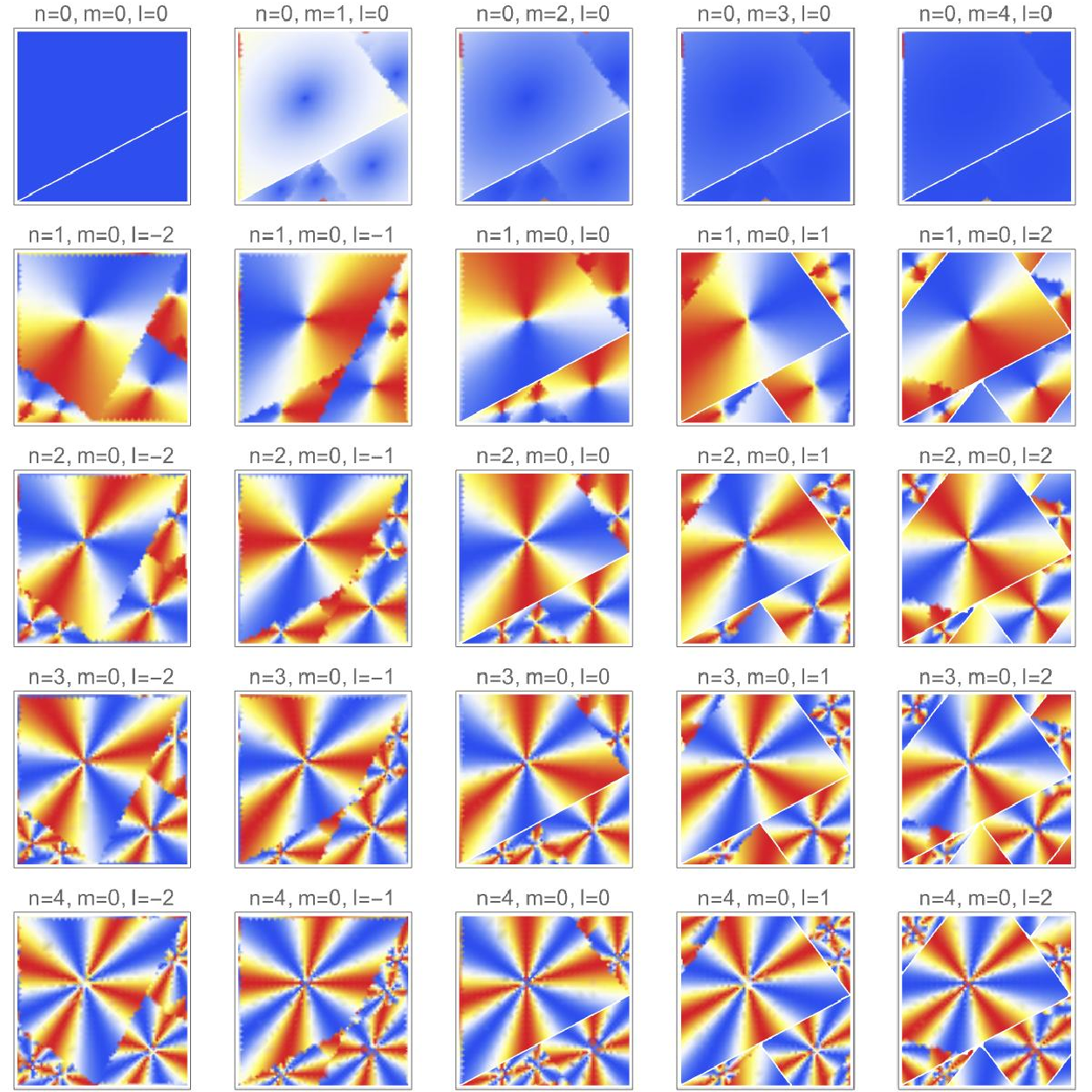}\\
          \caption{\label{KoopModescat2} The real part of the Koopman modes $\RE(f_{nml}(\theta))$ eq. (\ref{solcat}) for different values of $n$, $m$ and $l$, for the cat system defined by $\Phi = \left(\begin{array}{cc} 2\cos \omega & 1 \\ -1 & 0 \end{array}\right)$ with $\frac{\omega}{2\pi}=\frac{\sqrt 2}{2}$.}
        \end{center}
          \end{figure}
These ones reflect the fact that stable observables are waves following the rotation of the orbits around the fixed and cyclic points associated with the first return times in the patches. $n$ is the ``wave number'' of the oscillations around the special points, $m$ is the power of the amplitude growing of the waves with respect to the distance to the special points, and $l$ selects the special points around which the waves oscillate (with $l=0$ we have only the special points in the center of the coloured regions of right fig. \ref{cat2patches}, for $l\not=0$ we have the images by $\varphi^l$ of these points).
    
\end{example}

\subsection{Critical cat map case}
We call ``critical cat map'' the case where $\Phi$ is not diagonalisable and then the theorem \ref{ThFullKoopMode} does not apply. $\Phi$ takes the form of a Jordan block:
\begin{equation}
  \Phi = \left(\begin{array}{cc} 1 & 1 \\ 0 & 1 \end{array}\right)_{(e_a)}
\end{equation}
the degenerate eigenvalue being $1$ since $\det \Phi=1$, with $e_0$ the single eigenvector $\Phi e_0 = e_0$ and $e_* \in \ker(\Phi-\id)^2 \iff \Phi e_* = e_0+e_*$.
\begin{equation}
  \Phi^n = \left(\begin{array}{cc} 1 & n \\ 0 & 1 \end{array}\right)_{(e_a)}
\end{equation}
It follows that
\begin{eqnarray}
  \varphi^n(\theta) & = & \left(\begin{array}{c} e^1_0 (\theta^0+n\theta^*)+e^1_* \theta^* \\ e^2_0(\theta^0+n\theta^*) + e^2_* \theta^* \end{array} \right) \mod \left(\begin{array}{c} 2\pi \\ 2\pi \end{array}\right) \\
  & = & \theta + n e_0 \theta^* \mod (2\pi,2\pi)
\end{eqnarray}
with $\theta^* = e^*_\mu \theta^\mu$. The dynamics induced by $\varphi$ is a discrete translation in the $e_0$ direction over a distance depending on the value of $\theta^*$.\\

The ergodic components are $\Gamma_a = \{a e_*+ t e_0 \mod (2\pi,2\pi)\}_{t\in [0,t_*]}$, $\forall a \in [0,a_0]$, with $t_* = \frac{2\pi e_*^2}{e_0^1e_*^2-e_0^2e_*^2}$ and $a_0=\frac{2\pi e_0^2}{e_0^1e_*^2-e_0^2e_*^2}$ see fig. \ref{TransDiscr}.
\begin{figure}
  \begin{center}
    \includegraphics[width=5cm]{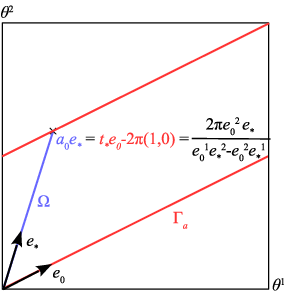}
    \caption{\label{TransDiscr} Representation of $\Gamma_a$ (for $a$ in the neighbourhood of $0$) and $\Omega = [0,a_0] \vec e_*$ for a discrete translation map induced by $\Phi$ non-diagonalisable.}
  \end{center}
\end{figure}
For almost all $\theta \in \Gamma_a$, $\overline{\{\varphi^n(\theta)\}_{n \in \mathbb Z}} = \overline{\{\theta+n a e_0\}_{n \in \mathbb Z}} = \Gamma_a$.\\

With $\Phi \in SL(2,\mathbb Z)$, $\frac{e_0^1}{e_0^2} = \frac{p}{r} \in \mathbb Q$ (with $p$ and $r$ coprimes). $\varphi$ presents a lot of cyclic points of any period $N$. Indeed let $(q^1,q^2) \in \mathbb Z$ such that $q^1 r = q^2 p$. $\theta$ such that $\theta^* = \frac{2q^1 \pi}{Ne_0^1} = \frac{2q^2 \pi}{N e_0^2}$ is $N$-cyclic, since by construction $\theta^\mu + N \theta^* e^\mu_0 - 2q^\mu \pi = \theta^\mu \iff \varphi^N(\theta)=\theta$. $\varphi$ presents then a sensitivity to initial conditions, nevertheless this one occur only for perturbations in the direction $e_*$.\\
A straight strip of any width directed by $e_0$ is invariant by $\varphi$, the flow is then not mixing. But if $A,B \in \mathscr B(\Gamma)$ are such that $\nu_L(e^*A \cap e^* B) \not=0$ (with $\nu_L$ the Lebesgue measure onto $[0,a_0]$ and $e^*$ viewed as the projection onto the line directed by $e_*$ along $e_0$) then $\exists n$, $\varphi^n(B) \cap A \not= \varnothing$. $\varphi$ is topologically ``half-transitive'' (it is mixing in the direction $e_0$ and non-mixing in the direction $e_*$).\\
The case of $\omega=0$ with non-diagonalisable $\Phi$, as limit case between the cases $\omega \in \mathbb R$ and $\omega \in \imath \mathbb R$, is a transition to the chaos ($\omega \in \imath \mathbb R^*$ induced a chaotic map, see below); $\varphi$ being close to a chaotic map (mixing in the direction $e_0$ and sensitive to initial conditions in the direction $e_*$) but not truly chaotic (because it is not completely transitive).\\

The space of observables is:
\begin{equation}
  L^2(\Gamma,d\mu) = \int_{[0,a_0]}^\oplus L^2\left(\Gamma_a,\frac{d\vartheta}{2\pi}\right) \frac{da}{a_0}
\end{equation}
where $\vartheta \in [0,2\pi[$ is a local coordinate onto the $\Gamma_a$ (which is a cycle since $\frac{e_0^1}{e_0^2} \in \mathbb Q$). If $a=\frac{2q^1 \pi}{Ne_0^1} = \frac{2q^2 \pi}{N e_0^2}$ with $q^1 r = q^2 p$, then $\Sp_{pp}(\left. \mathcal K \right|_{\Gamma_a}) = \{e^{\imath \frac{2\pi}{N} n}\}_{n \in \{0,...,N-1\}}$. But these eigenvalues for $n\not=0$ are not synthetic since $\nu_L(\Omega_{\frac{2\pi}{N}n})=0$ (because $\Omega_{\frac{2\pi}{N}n} = \{\frac{2q^1 \pi}{Ne_0^1}\}_{q^1\in \mathbb Z \cap \frac{p}{r} \mathbb Z}$ is countable). For the generic situation the ergodic Koopman modes of $\Gamma_a$ are $\forall e^{\imath \lambda} \in U(1)$:
\begin{equation}
  f_{\lambda a}(\theta) = \sum_{q=-\infty}^{+\infty} e^{\imath q\lambda} \delta_{\varphi^q(ae_*)}(\theta)
\end{equation}
On the whole of $\Gamma$, we can obtain Koopman modes associated with $e^{\imath \lambda}$ by $\int_{[0,a_0]}^\oplus c(a) f_{\lambda a}(\theta) \frac{da}{a_0}$, for any $L^2$ function $c$ such that $c(a_0)=c(0)$ since the two points correspond to the same ergodic component. We can then choose a Fourier basis:
\begin{equation}
  f_{\lambda n}(\theta) = \sum_{q=-\infty}^{+\infty} e^{\imath q \lambda} \int_0^{a_0} e^{\imath 2\pi n \frac{a}{a_0}} \delta_{\varphi^q(ae_*)}(\theta) \frac{da}{a_0}
\end{equation}
$1 \in \Sp_{pp}(\mathcal K)$ is of infinite multiplicity, with eigenfunctions $f_{0n}(\theta) = \int_0^{a_0} e^{\imath 2\pi n \frac{a}{a_0}} \delta_{\Gamma_a}(\theta) \frac{da}{a_0}$. For $\lambda \not=0$, $f_{\lambda a}$ is highly irregular (because $q(\theta)$ such that $ae_*+q(\theta)\theta^*e_0 = \theta \mod (2\pi,2\pi)$ for $\theta \in \Gamma_a$, strongly changes with $\theta$, $q(\theta)$ is discontinuous everywhere). It follows that $f_{\lambda n} \not\in L^2(\Gamma,d\mu)$ and then $\Sp_{pp}(\mathcal K) = \{1\}$ with $\Sp_{cont}(\mathcal K) = U(1) \setminus \{1\}$ (here all synthetic spectral values of the continuum of $\left. \mathcal K \right|_{\Gamma_a}$ are spectral values of the continuum of $\mathcal K$). This irregularity is due to the fact that any region of $\Gamma$ is driven by a lot of cyclic points of any periods.

\begin{example}
  The simplest example is $\Phi = \left(\begin{array}{cc} 1 & 1 \\ 0 & 1 \end{array} \right)$ for which the right and left Jordan generalised eigenvectors are
  \begin{eqnarray}
    e_0 = \left(\begin{array}{c} 1 \\ 0 \end{array}\right) & \qquad & e^0 = (\begin{array}{cc} 1 & -1 \end{array}) \\
    e_* = \left(\begin{array}{c} 1 \\ 1 \end{array}\right) & \qquad & e^* = (\begin{array}{cc} 0 & 1 \end{array})
  \end{eqnarray}
  In that case (by periodicity) $a_0=2\pi$. The ergodic components are
  \begin{equation}
    \Gamma_a  =  \{(a+t \mod 2\pi ,a)\}_{t \in [0,2\pi]} \quad \forall a \in ]0,2\pi]
  \end{equation}
  So $\Gamma_a$ is the fundamental circle of the torus $\theta^2=a$. The sensitivity to initial conditions of this system can be illustrated by fig. \ref{cat3SCI} and the ``half-mixing'' property by fig. \ref{cat3Cycles}.
   \begin{figure}
     \begin{center}
       \includegraphics[width=10cm]{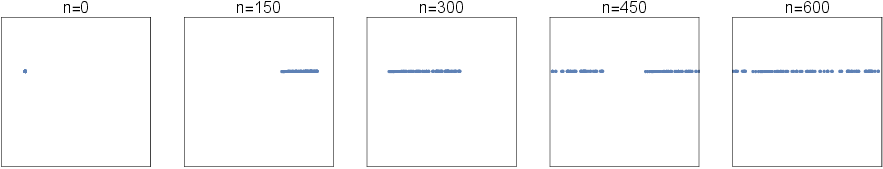}\\
       \caption{\label{cat3SCI} Representation in $\Gamma$ of some perturbations of an orbit $\{\varphi^n(\theta_0)\}_{n\in \mathbb N}$ at different iterations $n$, with $\Phi = \left(\begin{array}{cc} 1 & 1 \\ 0 & 1 \end{array}\right)$. The perturbed orbits are such that $\|\theta_0'-\theta_0\| \leq 10^{-2}$. The error grows linearly (the system is marginally stable) in the direction $e_0$ only.}
     \end{center}
   \end{figure}
   \begin{figure}
     \begin{center}
       \includegraphics[width=10cm]{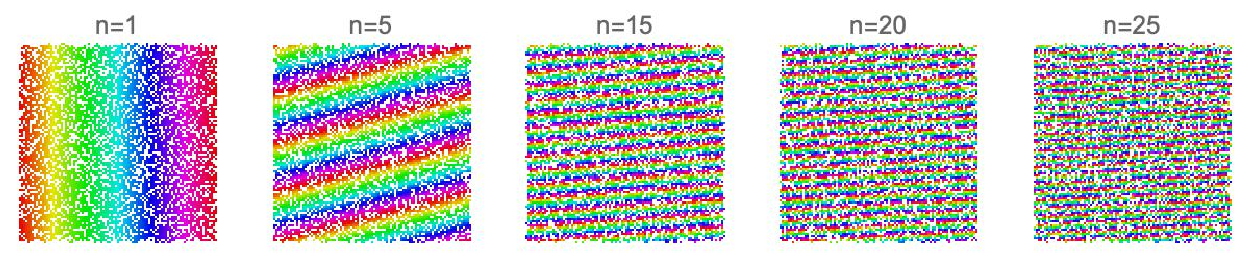}\\
       \caption{\label{cat3Cycles} Representation in $\Gamma$ of some orbits $\{\varphi^n(\theta_0)\}_{n\in \mathbb N}$ at different iterations $n$, with $\Phi = \left(\begin{array}{cc} 1 & 1 \\ 0 & 1 \end{array}\right)$. The orbits are coloured according to the value of $\theta_0^1$.}
     \end{center}
   \end{figure}
Since $\varphi^q(ae_*) = ae_* + qae_0 \mod (2\pi,0)$, the full Koopman modes are then
  \begin{eqnarray}
    f_{\lambda n}(\theta)  & = &\sum_{q=-\infty}^{+\infty} e^{\imath q \lambda} \int_0^{2\pi} e^{\imath n a} \delta(\theta^1-a(1+q) )\delta(\theta^2-a) \frac{da}{2\pi} \\
    & = & \int_0^{2\pi} \sum_{q=-\infty}^{+\infty} e^{\imath (q-1) \lambda+\imath n a} \delta(\theta^1-aq )\delta(\theta^2-a) \frac{da}{2\pi}
  \end{eqnarray}
  $\forall \theta^1$, $\forall \epsilon>0$, $\exists q,m_{\epsilon}(\theta^1) \in \mathbb Z$, such that $|\theta^1-qa+2m_{\epsilon}(\theta^1)\pi|<\epsilon$. The distribution for $\Gamma_a$ can be approximated by
  \begin{eqnarray}
    f_{\lambda a}^{(\epsilon)}(\theta) & = & e^{\imath (\frac{\theta^1+2m_{\epsilon}(\theta^1)\pi}{a}-1)\lambda} \delta(\theta^2-a) \\
    & = & e^{\imath (\frac{\theta^1}{a}-1)\lambda+ \imath \phi_{\epsilon a}(\theta^1)} \delta(\theta^2-a)
  \end{eqnarray}
  The value of the phase $\phi_{\epsilon a}(\theta^1) = \frac{2m_\epsilon(\theta^1)\pi}{a}$ changes strongly with $\theta^1$, taking the role of a pseudo-random phase. This fact is in relation with the ``half-mixing'' property of the flow inner to $\Gamma_a$. $f_{n\lambda}$ can be approximated by
  \begin{equation}
    f_{\lambda n}^{\mathfrak G}=\{e^{\imath (q_{\theta^1}-1)\lambda+\imath n \theta^2}\}_{\theta \in \mathfrak G}
  \end{equation}
  with $q_{\theta^1} \in \mathbb N$ such that $|\theta^1-(q_{\theta^1}\theta^2 \mod 2\pi)|<\epsilon_{\mathfrak G}$, where $\mathfrak G$ is a set of $n_{\mathfrak G}$ points uniformly randomly chosen in $\Gamma$ and $\epsilon_{\mathfrak G} = \frac{2\pi}{\sqrt{n_{\mathfrak G}}}$ (mean distance along one axis between two nearest neighbour points of $\mathfrak G$). Due to the pseudo-random phase, it is better to choose a Monte Carlo numerical representation of $f_{\lambda n}$ in place of a representation onto an uniform grid. This permits to avoid numerical artefacts in the representation. Some $f_{\lambda n}^{\mathfrak G}$ are represented fig. \ref{KoopModescat3v2}. 
      \begin{figure}
        \begin{center}
          \includegraphics[width=10cm]{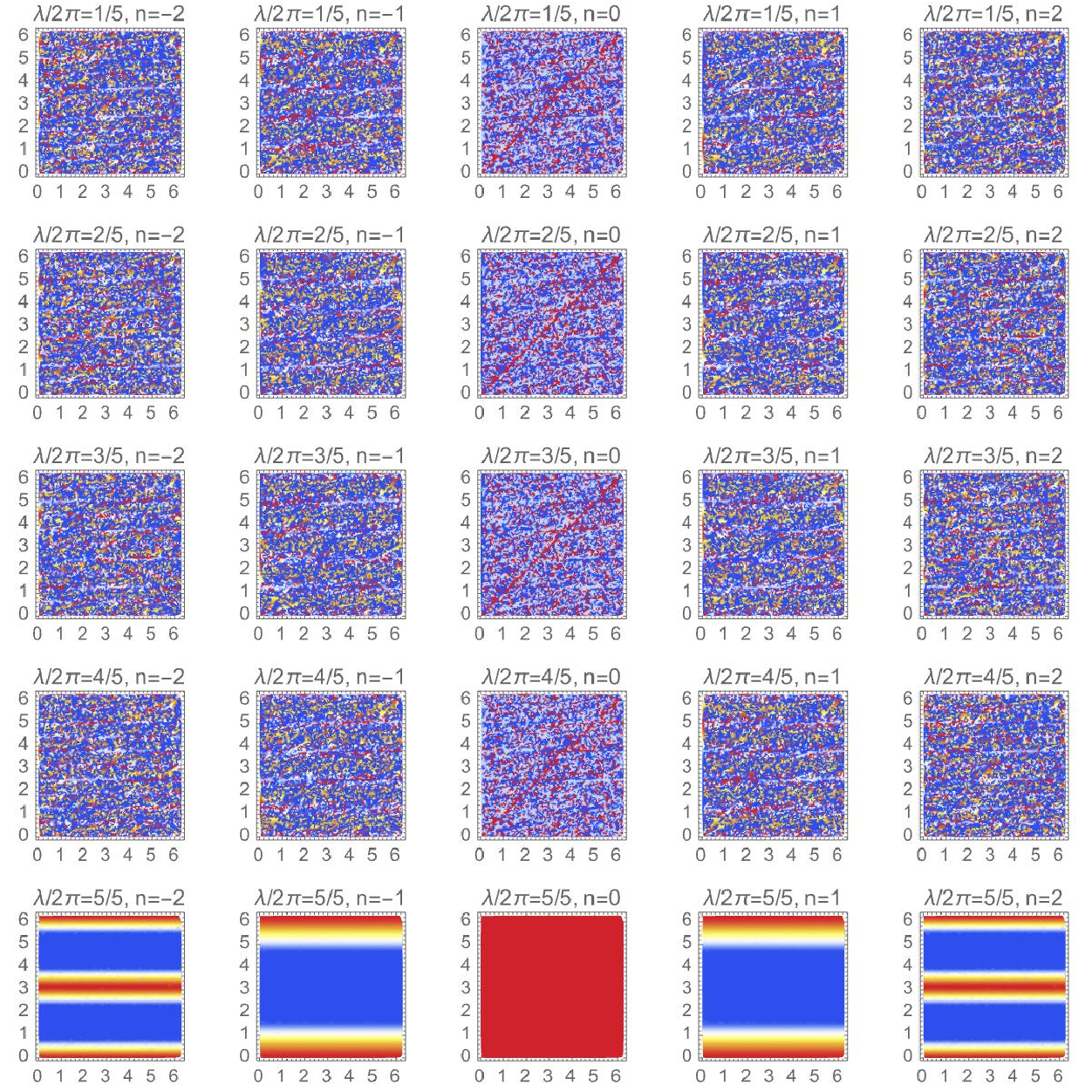}\\
          \caption{\label{KoopModescat3v2} The real part of the Koopman modes $\RE(f_{\lambda n}^{\mathfrak G}(\theta))$ with $n_{\mathfrak G}=10 000$ for the cat system defined by $\Phi = \left(\begin{array}{cc} 1 & 1 \\ 0 & 1 \end{array}\right)$. The case $\lambda=2\pi \iff \lambda=0$ corresponds to the eigenfunction of $\mathcal K$ associated with the single eigenvalue $1$.}
        \end{center}
      \end{figure}
   Except for $\lambda=0$, the stable observables are then close to noised signals, but these ones are not white noises (the oscillations of the cases $\lambda=0$ seems in part present for $\lambda \not=0$). A kind of regularity seems to occur in the noise. To enlighten this one, it is more useful to consider the Koopman distributions for a single $q$:
   \begin{eqnarray}\label{modescat3}
     f_{\lambda n q}(\theta) & = & e^{\imath (q-1) \lambda} \int_0^{2\pi} e^{\imath n a} \delta(\theta^1-qa)\delta(\theta^2-a)\frac{da}{2\pi} \\
     & = & e^{\imath (q-1)\lambda + \imath n \theta^2} \delta(\theta^1-q\theta^2)
   \end{eqnarray}
   $f_{\lambda n}$ being the superposition of all $f_{\lambda n q}$. These distributions represent the coherent modes in the dynamics induced by $\varphi$ because the pseudo-random phase is fixed to a single value along their support lines $\theta^1=q\theta^2$. $f_{\lambda n q}$ are represented fig. \ref{KoopModescat3}.
      \begin{figure}
        \begin{center}
          \includegraphics[width=10cm]{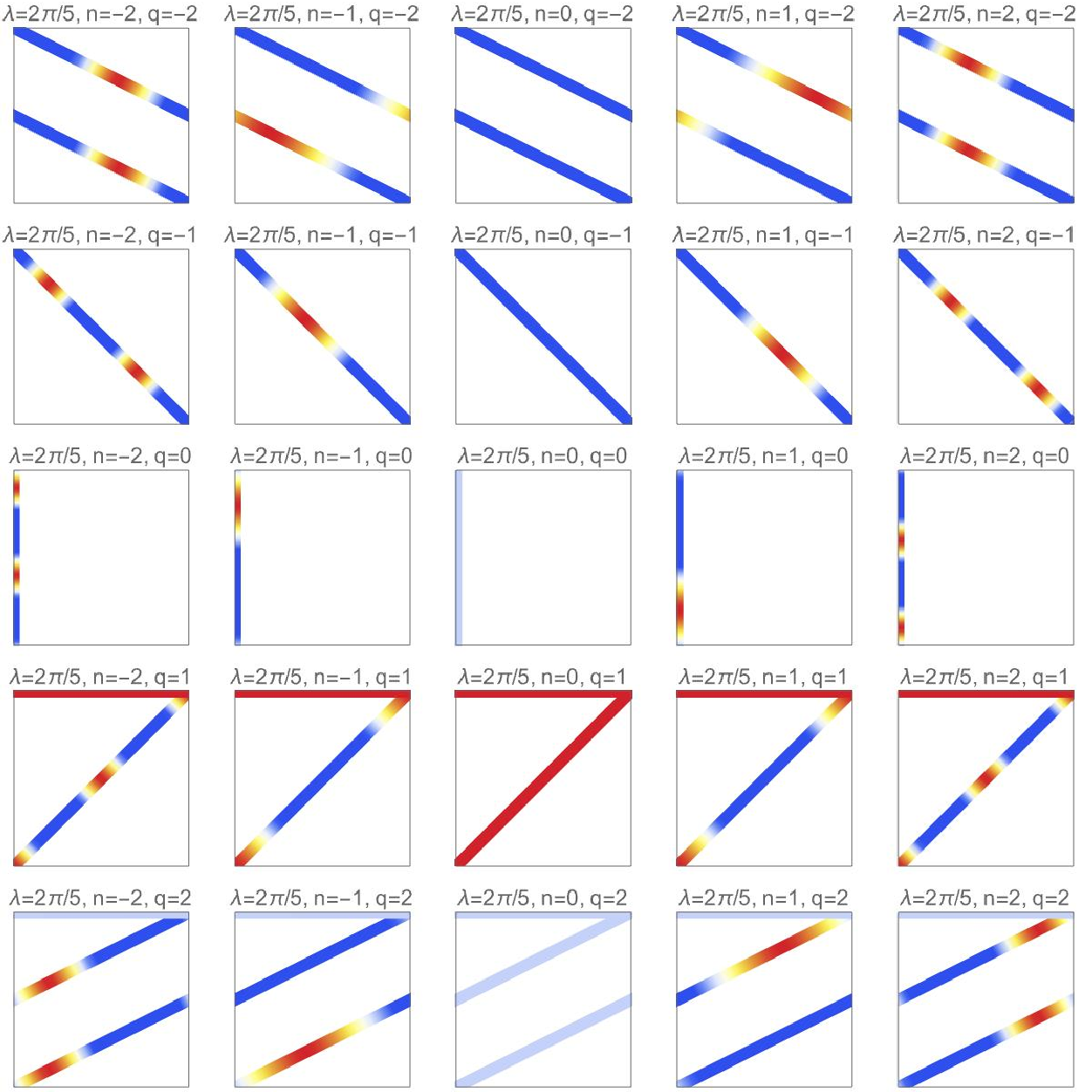}\\
          \caption{\label{KoopModescat3} The real part of the Koopman modes $\RE(f_{\lambda n q}(\theta))$ eq. (\ref{modescat3}) for the cat system defined by $\Phi = \left(\begin{array}{cc} 1 & 1 \\ 0 & 1 \end{array}\right)$.}
        \end{center}
      \end{figure}
$q$ determinates the direction of the coherent lines in $\Gamma$, $n$ is the wave number of the oscillations along these lines, and $\lambda$ (with $q$) determines the initial phase of these oscillations.     
\end{example}

\subsection{Chaotic cat map case}
We consider the case where $\omega = -\imath \pmb \lambda$ where $\pmb \lambda >0$ is the main Lyapunov exponent (see \cite{Goldstein}) of the cat map.
\begin{equation}
  \Phi =  \left(\begin{array}{cc} e^{\pmb \lambda} & 0 \\ 0 & e^{-\pmb \lambda} \end{array} \right)_{(e_a)}
\end{equation}
$\varphi$ is an Anosov flow (\cite{Lasota}), $e_+$ being the unstable direction whereas $e_-$ being the stable one. The system is mixing (and then ergodic) on the whole of $\Gamma$ for almost all initial condition $\theta_0$. There exists a dense set (but of zero measure) of cyclic points of all periods.\\

The Koopman modes associated with $e^{\imath \lambda}$ of the single ergodic component are the distributions:
\begin{equation}
  f_{\lambda \theta_0}(\theta) = \sum_{q=-\infty}^{+\infty} e^{\imath \lambda q} \delta_{\varphi^q(\theta_0)}(\theta)
\end{equation}
but due to the sensitivity to the initial conditions, $\|\varphi^q(\theta_0')-\varphi^q(\theta_0)\| \sim e^{q \pmb \lambda} \|\theta_0'-\theta_0\|$ (``$\sim$'' stands for the linearised behaviour), $f_{\lambda \theta_0}$ is then highly depend on $\theta_0$ and so we have a different Koopman mode for $\mu$-almost all initial condition $\theta_0 \in \Gamma$. A better approach can consist to use the expression of the full Koopman modes: $\forall \theta \in D_i$
\begin{equation}\label{KoopModeChaos}
  f_{\lambda \alpha l}(\theta) =  e^{\imath \frac{\lambda}{\pmb \lambda} \left(\alpha \ln (\varphi^l(\theta)^+-\phi^+_{i_{\varphi^l(\theta)}}) - (1-\alpha) \ln (\varphi^l(\theta)^--\phi^-_{i_{\varphi^l(\theta)}})\right)}
\end{equation}
$\alpha \in \mathbb R$ is not restricted, inducing that $\lambda$ is of uncountable infinite ``multiplicity''. This expression, which might appear to be that of a regular function, is misleading. Indeed we recall that
\begin{equation}
  \phi_{i_\theta}^\pm = 2\pi \frac{\sum_{q=0}^{N_\theta-1} e^{\pm (N_\theta-1-q)\pmb \lambda}m^\pm_{\imath_{\varphi^q(\theta)}}}{1-e^{\pm N_{\theta} \pmb \lambda}}
\end{equation}
where $N_\theta$ is the time of the first return of the flow starting at $\theta$ in $D_{i_\theta}$. But due to the sensitivity to initial conditions, $N_\theta$ is strongly dependent on $\theta$, inducing that the parameters $\phi_{i_\theta}^\pm$ appears as pseudo-random values. $f_{\lambda \alpha l}$ is then a highly irregular function. $f_{\lambda \alpha l} \not\in L^2(\Gamma,d\mu)$ since $f_{\lambda \alpha l}$ is not $\mu$-measurable.\\

 Due to the sensitivity to initial conditions, in the neighbourhood of $\theta$, we find an infinity of cyclic points $\phi_{i_{\theta'}}$ (for $\theta'$ in the neighbourhood of $\theta$) of all periods. In contrast with the cases $\omega \in \mathbb R^*$, where all point is cyclic or under the influence of a single cyclic orbit, $\mu$-almost all points are here under the influence of an infinity of cyclic orbits of all period. This is precisely this property which induces the chaotic behaviour of the orbits. For the Koopman representation, this is this property which is responsible for the irregularity of the Koopman modes.\\

 The two expressions of the Koopman modes (ergodic and full) are related as follows. $\sum_q f_{\lambda \alpha l}(\theta) \delta_{\varphi^q(\theta_0)}(\theta) = f_{\lambda \alpha l}(\theta_0) \sum_q e^{\imath q \lambda} \delta_{\varphi^q(\theta_0)} = e^{\imath \zeta_{\alpha l \theta_0}} f_{\lambda \theta_0}(\theta)$ (with $f_{\lambda \alpha l}(\theta_0)=e^{\imath \zeta_{\alpha l \theta_0}}$).\\

 Remark: consider the following change of Koopman modes:
\begin{eqnarray}
  \int_{-\infty}^{+\infty} A_{\lambda l}(\alpha) f_{\lambda \alpha l}d\alpha & = & e^{-\imath \frac{\lambda}{\pmb \lambda} \xi^-} \int_{-\infty}^{+\infty} A_{\lambda l}(\alpha) e^{\imath \frac{\lambda}{\pmb \lambda} \alpha(\xi^++\xi^-)} d\alpha \\
    & = & \sqrt{2\pi} e^{-\imath \frac{\lambda}{\pmb \lambda} \xi^-} \hat A_{\lambda l}\left(\frac{\lambda}{\pmb \lambda}(\xi^++\xi^-)\right)
\end{eqnarray}
where $\xi^\pm = \ln(\varphi^l(\theta)^\pm-\phi^\pm_{i_{\varphi^l(\theta)}})$ and $\hat A_{\lambda l}$ is the Fourier transformation of $A_{\lambda l}$. We have well
\begin{equation}
  \mathcal K e^{-\imath \frac{\lambda}{\pmb \lambda} \xi^-} \hat A_{\lambda l}\left(\frac{\lambda}{\pmb \lambda}(\xi^++\xi^-)\right) = e^{\imath \lambda} e^{-\imath \frac{\lambda}{\pmb \lambda} \xi^-} \hat A_{\lambda l}\left(\frac{\lambda}{\pmb \lambda}(\xi^++\xi^-)\right)
\end{equation}
since $\xi^\pm \circ \phi = \xi^\pm \pm \pmb \lambda$, for any $\hat A_\lambda \circ \frac{\lambda}{\pmb \lambda}(\xi^++\xi^-)$ square integrable or not, or even a singular distribution.\\

$\Sp(\mathcal K) = U(1)$, $e^0=1$ being the single pure point eigenvalue which is not degenerate. For a $N$-cyclic orbit $\Orb(\theta_*)$ we have $\Sp_{pp}(\left. \mathcal K \right|_{\Orb(\theta_*)}) = \{e^{\imath \frac{2\pi n}{N_{\theta_*}}}\}_{n \in \{0,...,N_{\theta_*}-1\}}$, but these eigenvalues are not synthetic.\\

\begin{example}
  We consider the Arnold's cat map defined by $\Phi = \left(\begin{array}{cc} 1 & 1 \\ 1 & 2 \end{array} \right)$. The mean Lyapunov exponent is $\pmb \lambda = \ln \frac{3+\sqrt 5}{2}$ with the eigenvectors:
  \begin{equation}
    e_\pm = \frac{1}{\sqrt{\mathfrak g_\pm+2}} \left(\begin{array}{c} 1 \\ \mathfrak g_\pm \end{array} \right) \qquad \text{with }\mathfrak g_\pm = \frac{1\pm\sqrt 5}{2}
  \end{equation}
  $[0,2\pi[^2$ is divided in four patches: $D_1=\{(\theta^1,\theta^2)\in [0,2\pi[^2, \theta^2 \leq -\frac{\theta^1}{2}+\pi\}$, $D_2=\{(\theta^1,\theta^2)\in [0,2\pi[^2, -\frac{\theta^1}{2}+\pi \leq \theta^2 \leq -\theta^1+2\pi\}$, $D_3=\{(\theta^1,\theta^2)\in [0,2\pi[^2, -\theta^1+2\pi \leq \theta^2 \leq -\frac{\theta^1}{2}+2\pi\}$ and $D_4=\{(\theta^1,\theta^2)\in [0,2\pi[^2, \theta^2 \geq -\frac{\theta^1}{2}+2\pi\}$; the associated translation are $m_1=(0,0)$, $m_2=(0,-1)$, $m_3=(-1,-1)$ and $m_4=(-1,-2)$; see fig. \ref{cat5patches}.
 \begin{figure}
   \begin{center}
     \includegraphics[width=5cm]{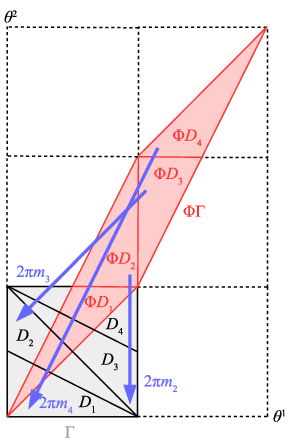}
     \caption{\label{cat5patches} Graphical representation of $\varphi: \theta \mapsto \Phi \theta + 2\pi m_i$, with $\Phi = \left(\begin{array}{cc} 1 & 1 \\ 1 & 2 \end{array} \right)$. $\Gamma$ is represented by $[0,2\pi[^2$.}
   \end{center}
 \end{figure}

 Due to the positive mean Lyapunov exponent, this system is sensitive to initial conditions as illustrated by fig. \ref{cat5SCI}.
   \begin{figure}
     \begin{center}
       \includegraphics[width=10cm]{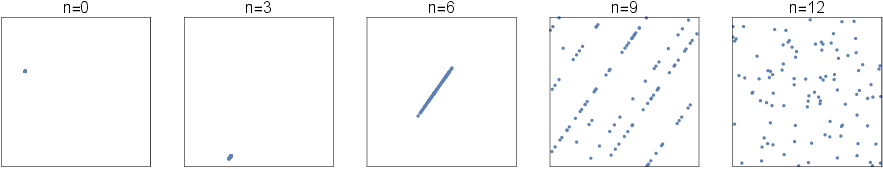}\\
       \caption{\label{cat5SCI} Representation in $\Gamma$ of some perturbations of an orbit $\{\varphi^n(\theta_0)\}_{n\in \mathbb N}$ at different iterations $n$, with $\Phi = \left(\begin{array}{cc} 1 & 1 \\ 1 & 2 \end{array}\right)$. The perturbed orbits are such that $\|\theta_0'-\theta_0\| \leq 10^{-2}$. The errors grow exponentially in accordance with the positive mean Lyapunov exponent.}
     \end{center}
   \end{figure}
   The system is also mixing, as illustrated fig. \ref{cat5Cycles}.
     \begin{figure}
     \begin{center}
       \includegraphics[width=10cm]{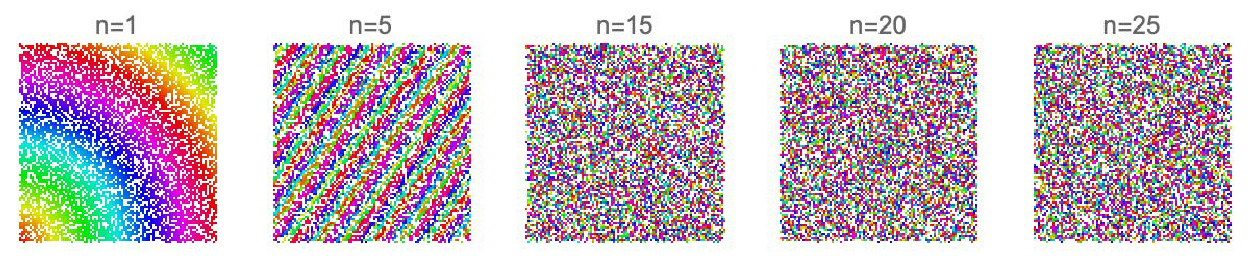}\\
       \caption{\label{cat5Cycles} Representation in $\Gamma$ of some orbites $\{\varphi^n(\theta_0)\}_{n\in \mathbb N}$ at different iterations $n$, with $\Phi = \left(\begin{array}{cc} 1 & 1 \\ 1 & 2 \end{array} \right)$. The orbits are coloured according to the distance of $\theta_0$ from $(0,0)$.}
     \end{center}
     \end{figure}

The Koopman modes $f_{\lambda \theta_0}$ are similar to white noises (see fig. \ref{KoopModescat5v2}).
     \begin{figure}
        \begin{center}
          \includegraphics[width=10cm]{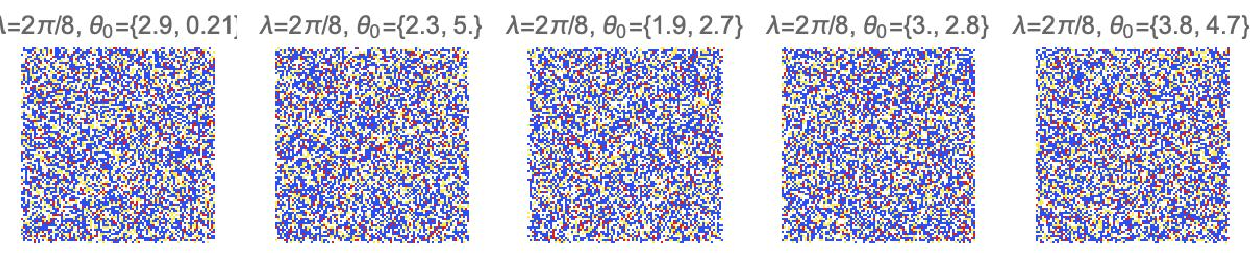}\\
          \caption{\label{KoopModescat5v2} The real part of the Koopman modes $\RE(f_{\lambda \theta_0}(\theta))$ estimated with orbits of 10 000 points for the Arnold's cat map defined by $\Phi = \left(\begin{array}{cc} 1 & 1 \\ 1 & 2 \end{array} \right)$.}
        \end{center}
      \end{figure}
    The Koopman modes $f_{\lambda \alpha l}$ can be approximated by
\begin{equation}
  f_{\lambda \alpha l}^{\mathfrak G} = \{f_{\lambda \alpha l}(\theta)\}_{\theta \in \mathfrak G}
\end{equation}
where $\mathfrak G$ is a set of $n_{\mathfrak G}$ points uniformly randomly chosen in $\Gamma$, the time of return of each point $\theta \in \mathfrak G$ is computed by iteration of the dynamics. These ones are represented fig. \ref{KoopModescat5}.
     \begin{figure}
        \begin{center}
          \includegraphics[width=10cm]{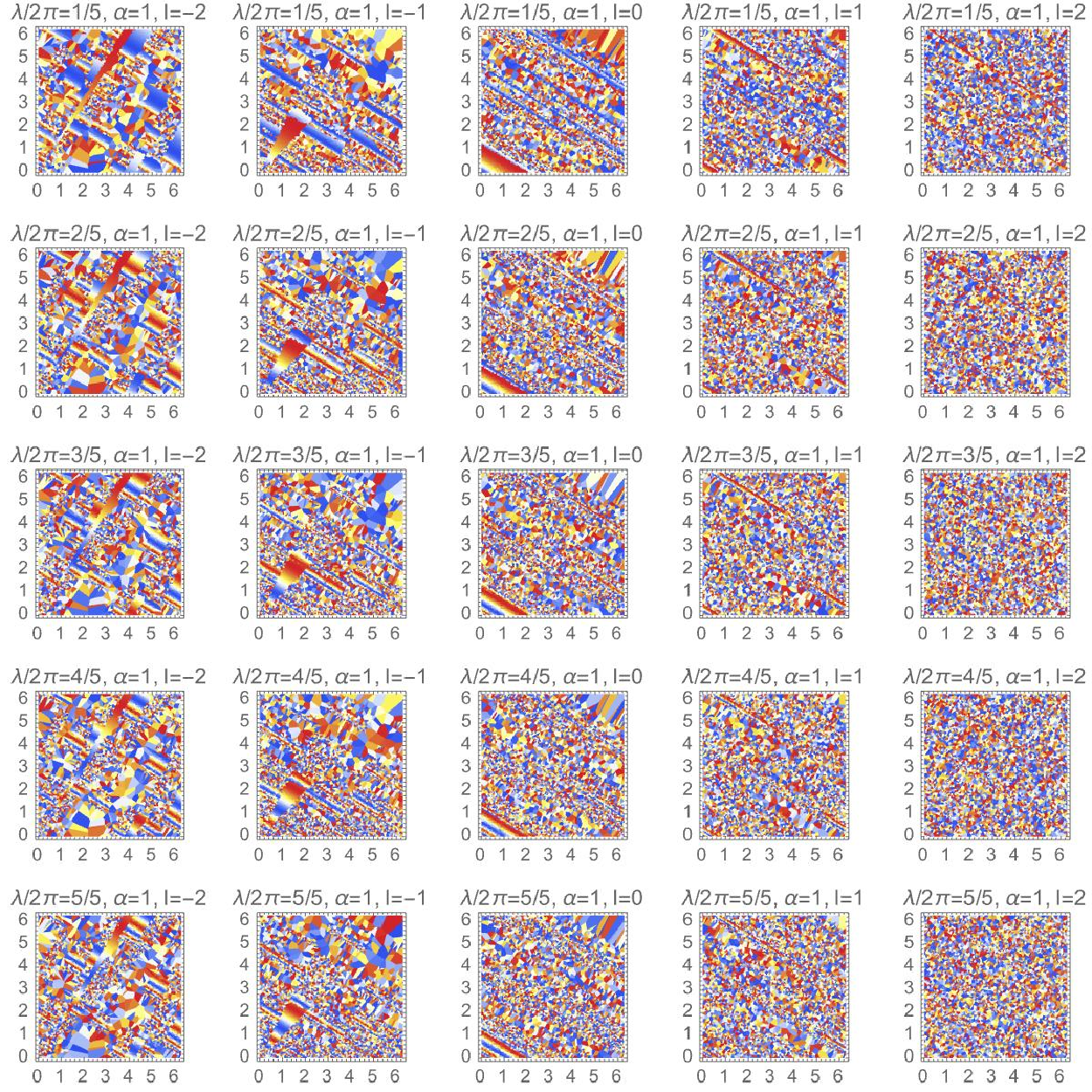}\\
          \caption{\label{KoopModescat5} The real part of the Koopman modes $\RE(f_{\lambda \alpha l}^{\mathfrak G}(\theta))$ with $n_{\mathfrak G}=10000$ for the Arnold'cat map defined by $\Phi = \left(\begin{array}{cc} 1 & 1 \\ 1 & 2 \end{array} \right)$.}
        \end{center}
     \end{figure}
 Anew the plots seem to be essentially noises due to the pseudo-random character of the first return time. The roles of $\lambda$, and $\alpha$, $l$ (for the Full Koopman mode expression) or $\theta_0$ (for the ergodic Koopman mode expression) are not obvious. Due to the pseudo-random character of the Koopman modes, they are just parameters ``initialising'' the pseudo-random generator of white noises. The strong chaoticity of the Arnold's cat map induces that the only stable observables are undistinguishable from white noises. Nevertheless the expression eq. \ref{KoopModeChaos} show that these white noises are not truly random and are in fact deterministic. This is the same thing than for the orbits fig. \ref{cat5Cycles}, they appear as white noises whereas the dynamics is deterministic.
\end{example}

\subsection{Summary}
The properties of cat dynamical systems are summarised table \ref{cat}.
\begin{table}
  \caption{\label{cat} Spectral properties of a classical cat map $\varphi(\theta) = \Phi \theta \mod (2\pi,2\pi)$, with respect to $\Sp(\Phi) = \{e^{\pm \imath \omega}\}$.  $\omega=0$ concerns only the cases where $\Phi$ is not diagonalisable and not the trivial case $\Phi=\id$. ``spp'' stands for ``synthetic pure point'', ``npp'' for ``non-synthetic pure point'', ``scont'' for ``synthetic continuous'' and ``ncont'' for ``non-synthetic continuous''.\\
    $\Gamma_a$ stands for a generic ergodic component (since the cat map is ergodic on the whole of $\Gamma$ for $\omega \in \imath \mathbb R^*$, the ergodic decomposition does not hold in this case). $\theta_*$ are some ``exceptional'' points which are $\mu$-almost nowhere in $\Gamma$.}
  \begin{tabular}{r|c|c|c|c}
    & $\frac{\omega}{2\pi} \in \mathbb Q^*$ & $\frac{\omega}{2\pi} \in \mathbb R \setminus \mathbb Q$ & $\omega=0$ & $\omega \in \imath \mathbb R^*$ \\
    \hline\hline
    cyclicality of $\varphi$ & cyclic & quasi-cyclic & critical & chaotic \\
    $\varphi$ mixing & no & no & half & yes \\
    S.I.C. of $\varphi$ & no & no & yes & yes \\
    \hline
    $\Sp_{spp}(\left. \mathcal K \right|_{\Gamma_a})$ & $\{e^{\imath \frac{2\pi n}{N}}\}_{n<N}$ & $\{1\}$ & $\{1\}$ &  \\
    $\Sp_{npp}(\left. \mathcal K \right|_{\Orb(\theta_*)})$ & $\varnothing$ & $\{e^{\imath \frac{2\pi n}{N_*}}\}_{n<N_*}$ & $\{e^{\imath \frac{2\pi n}{N_*}}\}_{n<N_*}$ & $\{e^{\imath \frac{2\pi n}{N_*}}\}_{n<N_*}$  \\
    $\Sp_{scont}(\left. \mathcal K \right|_{\Gamma_a})$ & $\varnothing$ & $U(1)\setminus\{1\}$ &  $U(1)\setminus\{1\}$ & \\
    $\Sp_{ncont}(\left. \mathcal K \right|_{\Gamma_a})$ & $\varnothing$ & $\varnothing$ & $\varnothing$ &  \\
    \hline
    $\Sp_{pp}(\mathcal K)$ & $\{e^{\imath \frac{2\pi n}{N}}\}_{n<N}$ & $\{e^{\imath n \omega}\}_{n \in \mathbb Z}$ & $\{1\}$ & $\{1\}$ \\
    $\Sp_{cont}(\mathcal K)$ & $\varnothing$ & $U(1) \setminus \{e^{\imath n \omega}\}_{n \in \mathbb Z}$ & $U(1) \setminus \{1\}$ & $U(1) \setminus \{1\}$\\
    multiplicity of $1$ & $\infty$ & $\infty$ & $\infty$ & $1$ 
  \end{tabular}
\end{table}
The regular dynamics are characterised by a non-trivial Koopman pure point spectrum, the cyclic case differing from the quasi-cyclic case by a finite number of eigenvalues. The erratic dynamics are characterised by a trivial Koopman pure point spectrum, the chaotic case differing from the critical one by the non-degenerescence of $1$ (implying a flow ergodic on the whole of $\Gamma$). The Koopman spectra by the ergodic decomposition are similar in the quasi-cyclic and the critical cases, in accordance with their ``position'' of transition to chaos.\\
The full Koopman modes of the regular dynamics appear as ``waves'' in $\Gamma$, whereas those of the erratic dynamics appear as ``noises''. This is in accordance with the interpretation of the full Koopman modes as stable observables defined coherently on the whole of $\Gamma$. For the transitional cases, the waves of the quasi-cyclic maps are strongly cut following a puzzle of coherent regions in $\Gamma$; whereas the noises of the critical maps hide a little regularity along some lines in $\Gamma$.

\section{Conclusion}
We have seen that full Koopman modes, representing stable classical observables coherently defined on the whole phase space, can be built for classical cat maps. These full Koopman modes are strongly related to the special orbits (fixed points and $N$-cyclic orbits) in the complexified phase space. Full Koopman modes are continuous functions in the phase space regions driven by a single cyclic point; they are highly irregular functions in regions driven by a lot of cyclic points of any periods. These full Koopman modes can be naturally decomposed into Koopman modes associated with the ergodic components. But the converse is not obvious, the knowledge of Koopman modes on each ergodic components does not permit to built coherent full Koopman modes (because of the synthetic spectral values of the continuum which can be pure point eigenvalues or spectral value of the continuum of the synthetised Koopman operator). Moreover in contrast with the Koopman modes of the ergodic decomposition, the full Koopman modes form countable basis of the eigenspaces (even if the spectrum is continuous as for the quasi-cyclic cat maps) except for the highly irregular cases (critical and chaotic cat maps). \\

The generalisation of the present results to another classes of systems is an interesting question. It is better to be cautious about the matter, because cat maps have a very special property: they are quasilinear (the only nonlinearity is the modulo $2\pi$, implying that the Jacobian matrix of the flow is the same on the whole phase space). The quasilinearity plays a significant role in the production of the present results, since it permits to relate the Koopman modes to the cyclic points with the eigenvalues and the eigenvectors of $\Phi$. The same thing does probably not occur with the Jacobian matrices of more nonlinear flows.

\end{document}